\documentclass[aps,prd,superscriptaddress,nofootinbib,preprintnumbers,
amsmath,amssymb]{revtex4-2}

\usepackage[english]{babel}
\usepackage{amsmath,amssymb,amsthm}
\usepackage{array}
\usepackage{booktabs}
\usepackage{float}
\usepackage{graphicx}
\usepackage{xcolor}
\usepackage{tikz}
\usetikzlibrary{arrows.meta,positioning,shapes.geometric}
\usepackage[colorlinks=true,citecolor=blue,urlcolor=blue,linkcolor=blue]{hyperref}

\graphicspath{{figures/}}
\newcommand{\beq}{\begin{equation}}
\newcommand{\eeq}{\end{equation}}
\newcommand{\bea}{\begin{eqnarray}}
\newcommand{\eea}{\end{eqnarray}}

\newcommand{\comm}[1]{}

\newtheorem{proposition}{Proposition}
\newtheorem{theorem}{Theorem}

\begin{document}

\title{Closure of Traceful Radiation-Gauge Reconstruction with Matter:
General Kerr Conditions, a Thin-Shell Benchmark, and a Conserved Torus}
\author{G.~G.~L. Nashed}
\email{nashed@bue.edu.eg}
\affiliation{Centre for Theoretical Physics, The British University in Egypt,
P.O. Box 43, El Sherouk City, Cairo 11837, Egypt}
\affiliation{Centre for Space Research, North-West University,
Potchefstroom 2520, South Africa}
\date{\today}

\begin{abstract}

We formulate the closure problem for first-order, traceful radiation-gauge
metric reconstruction on a Kerr background with matter. The reconstruction is
decomposed into Weyl-curvature, Ricci-trace, and mass--angular-momentum
completion sectors. For compatible initial, gauge, and completion data, we
establish local existence and uniqueness of the resulting triangular transport
system on a regular Kerr domain, and identify the additional conditions needed
for a closed global formulation. This is not a global existence theorem. 
We test the framework using two conserved matter sources. For a thin shell
around a Schwarzschild black hole, Israel matching, regularity, asymptotic
flatness, a barotropic equation of state, and fixed rest mass determine the
static, axial, and leading rotational sectors. The ADM-mass shift is
$\sqrt{1-2M/R}$ times the shell rest mass, while the axial sector determines
the angular momentum, moment of inertia, and interior frame dragging. We also
construct a stationary constant-angular-momentum fluid torus on Kerr and
compute its Newman--Penrose Ricci projections, showing that nonspherical
sources generally excite coupled Ricci sectors. The shell provides a closed
benchmark, while the torus illustrates the additional structure required for
general Kerr sources.
\end{abstract}
 \maketitle
\section{Introduction}
Black-hole perturbation theory describes the response of compact objects to
weak disturbances and underlies gravitational-wave modelling, self-force
calculations, extreme-mass-ratio inspirals, and strong-field tests of general
relativity \cite{Poisson:2011nh,Barack:2018yvs}. Rotating black holes are
therefore naturally studied through perturbations of the Kerr geometry
\cite{Kerr:1963ud,Teukolsky:2014vca}.

The radiative degrees of freedom are encoded in the Teukolsky equations for
the extreme Weyl scalars $\Psi_0^{(1)}$ and $\Psi_4^{(1)}$
\cite{Teukolsky:1972my,Teukolsky:1973ha}. However, many applications require
the metric perturbation $h_{\mu\nu}$. In vacuum, the
Chrzanowski--Cohen--Kegeles construction reconstructs this metric from Hertz
potentials in a radiation gauge
\cite{Cohen:1974cm,Chrzanowski:1975wv,Kegeles:1979an,Wald:1978vm}.

With matter, the radiative Weyl scalars do not determine the Ricci sector
generated by the stress--energy tensor, and perturbations of the background
mass and angular momentum require a separate completion sector
\cite{Wald:1978vm,Ori:2002uv,Keidl:2010pm,Merlin:2016boc}. Li and Yunes
recently developed a traceful radiation-gauge reconstruction for sourced
perturbations of Petrov type-D spacetimes \cite{Li:2026rkf}. In the outgoing
gauge, the Ricci projection $\Phi_{22}^{(1)}$ drives a hierarchy of
first-order transport equations for $\mu^{(1)}$,
$h_{m\bar m}^{(1)}$, and the metric trace.

The main aim of this work is to formulate the closure conditions for sourced,
traceful radiation-gauge metric reconstruction on Kerr and to state precisely
the scope of the result. We separate the Weyl-curvature, Ricci-trace, and
mass--angular-momentum completion sectors, and prove local existence and
uniqueness of the triangular transport hierarchy on a regular Kerr domain,
modulo residual gauge and completion freedoms. Theorem~\ref{thm:local_reconstruction}
establishes this local result, while Proposition~\ref{prop:kerr_closure}
specifies the source, boundary, interface, charge, and kernel conditions
needed for a closed global problem. The latter is conditional and does not
constitute a global existence-and-uniqueness theorem.

We illustrate the formulation using two conserved sources. A thin shell around
a Schwarzschild black hole provides a controlled junction benchmark: its
static, axial, and quadratic rotational sectors determine the mass completion,
angular momentum, interior frame dragging, rotational displacement, and
quadrupole. A stationary constant-angular-momentum perfect-fluid torus on Kerr
provides a smooth nonspherical source and demonstrates the coupling of the
$\Phi_{20}^{(1)}$, $\Phi_{21}^{(1)}$, and $\Phi_{22}^{(1)}$ Ricci channels.

Section~\ref{II} introduces the Kerr background and tetrad, while
Secs.~\ref{III} and \ref{IV} specify the gauges and review the traceful
transport hierarchy \cite{Li:2026rkf}. Sections~\ref{Vclosed} and
\ref{sec:torus} develop the shell and torus applications. Sections~\ref{VI}
and \ref{VII} formulate the general sourced Kerr closure problem and clarify
the status of the resulting framework; Sec.~\ref{VIII} concludes.

\section{Kerr Background}\label{II}

The background spacetime considered in this work is the Kerr solution, which describes a stationary, axisymmetric, asymptotically flat vacuum black hole of mass $M$ and angular momentum $J=aM$ \cite{Kerr:1963ud,Carter1973,Chandrasekhar}. In Boyer--Lindquist coordinates $(t,r,\theta,\phi)$, the metric takes the form

\begin{equation}
ds^2=
-\left(1-\frac{2Mr}{\Sigma}\right)dt^2
-\frac{4Mar\sin^2\theta}{\Sigma}\,dt\,d\phi
+\frac{\Sigma}{\Delta}dr^2
+\Sigma d\theta^2
+\frac{A\sin^2\theta}{\Sigma}d\phi^2,
\end{equation}

where

\begin{equation}
\Delta=r^2-2Mr+a^2,
\qquad
\Sigma=r^2+a^2\cos^2\theta,\quad \mbox{and} \qquad A=(r^2+a^2)^2-a^2\Delta\sin^2\theta.
\end{equation}
The event horizon is located at
\begin{equation}
r_+=
M+\sqrt{M^2-a^2},\qquad \mbox{with $|a|\le M$.}
\end{equation}

To describe perturbations we employ the Newman--Penrose formalism \cite{Newman:1961qr}, based on a complex null tetrad $(l^\mu,n^\mu,m^\mu,\bar m^\mu)$ satisfying

\begin{equation}
l^\mu n_\mu=-1,
\qquad
m^\mu\bar m_\mu=1,
\end{equation}
with all remaining inner products vanishing. We adopt the Kinnersley tetrad \cite{Kinnersley:1969zza},

\begin{align}
l^\mu
&=
\left(
\frac{r^2+a^2}{\Delta},
1,
0,
\frac{a}{\Delta}
\right),
\
n^\mu
&=
\frac{1}{2\Sigma}
\left(
r^2+a^2,
-\Delta,
0,
a
\right),
\
m^\mu
&=
\frac{1}{\sqrt{2}(r+i a\cos\theta)}
\left(
i a\sin\theta,
0,
1,
\frac{i}{\sin\theta}
\right),
\end{align}
with $\bar m^\mu$ obtained by complex conjugation. In this tetrad, the non-vanishing Newman--Penrose spin coefficients are

\begin{align}
\rho
&=
-\frac{1}{r-i a\cos\theta},
\qquad
\mu
=
\frac{\Delta}{2\Sigma}\rho^2\bar\rho,
\qquad
\tau
=
-\frac{i a\sin\theta}{\sqrt{2}}\rho\bar\rho,
\qquad 
\pi
=
\frac{i a\sin\theta}{\sqrt{2}}\rho^2,
\qquad
\beta
=
-\frac{\bar\rho\cot\theta}{2\sqrt{2}},
\nonumber\\
\alpha
&=
\pi-\bar\beta,
\qquad
\gamma
=
\mu+\frac{r-M}{2\Sigma}, \quad \mbox{while} \quad
\kappa=\sigma=\lambda=\nu=\epsilon=0.
\end{align}
The vanishing of $\kappa$, $\sigma$, $\lambda$, and $\nu$ reflects the alignment of the tetrad with the repeated principal null directions of the Kerr geometry and is responsible for the simplifications underlying the Teukolsky formalism \cite{Teukolsky:2014vca,Teukolsky:1973ha,Chandrasekhar}. The Kerr spacetime is Petrov type D, and therefore all background Weyl scalars vanish except $\Psi_2$,
\begin{equation}
\Psi_2=
-\frac{M}{(r-i a\cos\theta)^3},\quad \mbox{or equivalently} \quad
\Psi_0
=
\Psi_1
=
\Psi_3
=
\Psi_4
=
0,
\qquad
\Psi_2\neq0.\label{psi2}
\end{equation}
This algebraic specialization forms the foundation of black-hole perturbation theory on Kerr backgrounds. In particular, perturbations of the extreme Weyl scalars $\Psi_0^{(1)}$ and $\Psi_4^{(1)}$ satisfy the Teukolsky master equation and provide the natural starting point for metric reconstruction \cite{Teukolsky:2014vca,Teukolsky:1973ha}.
\section{Traceful Radiation Gauge}

Throughout the formulation sections, a single bookkeeping parameter
$\lambda$ measures the strength of a generic source.  Thus
\begin{equation}
g_{\mu\nu}(\lambda)=g^{(0)}_{\mu\nu}
+\lambda h^{(1)}_{\mu\nu}+O(\lambda^2),
\qquad
T_{\mu\nu}(\lambda)=\lambda T^{(1)}_{\mu\nu}+O(\lambda^2),
\label{eq:generic_ordering}
\end{equation}
where $g^{(0)}_{\mu\nu}$ is the Kerr metric and every Newman--Penrose
quantity carrying a superscript $(1)$ is the coefficient of
$\lambda$.  No second-order metric terms are used in the general Kerr
reconstruction.  The shell application has an additional, independent
slow-rotation parameter; its bivariate ordering is defined explicitly
in Sec.~\ref{Vclosed}.  The torus application instead uses its total
energy divided by the black-hole mass as the single source-strength
parameter.  These conventions prevent ``first order'' in source
strength from being confused with first order in rotation.
\label{III}

Having established the Kerr background and the associated Newman--Penrose tetrad, we now introduce the gauge framework used in the reconstruction procedure. Since the reconstruction is naturally formulated with respect to the principal null directions of the Kerr spacetime, it is convenient to impose the gauge conditions directly on the tetrad components of the metric perturbation.

Let $h_{\mu\nu}$ denote a first-order perturbation of the Kerr metric $g_{\mu\nu}$. We impose the radiation-gauge condition

\begin{equation}
n^\mu h_{\mu\nu}=0,\quad \mbox{where $n^\mu$ is one of the principal null directions of the Kinnersley tetrad.}
\label{eq:traceful_RG}
\end{equation}
 This condition eliminates all tetrad components containing one projection along $n^\mu$,

\begin{equation}
h_{nn}
=
h_{nl}
=
h_{nm}
=
h_{n\bar m}
=
0.
\label{eq:RG_components}
\end{equation}
Throughout this work, tetrad components are defined by

\begin{equation}
h_{ab}=
h_{\mu\nu}
e_a^\mu e_b^\nu,
\qquad
e_a^\mu
\in
\{l^\mu,n^\mu,m^\mu,\bar m^\mu\}.
\end{equation}
In the conventional Chrzanowski--Cohen--Kegeles (CCK) reconstruction formalism, the radiation gauge is supplemented by the trace-free condition
\begin{equation}
h=0,\quad \mbox{thereby eliminating the trace component \cite{Cohen:1974cm,Chrzanowski:1975wv,Kegeles:1979an}.}
\end{equation}
 Following Li and Yunes~\cite{Li:2026rkf}, we relax this additional restriction and allow the metric trace

\begin{equation}
h
=
g^{\mu\nu}h_{\mu\nu}, \qquad \mbox{to remain nonvanishing. }
\label{eq:trace_definition}
\end{equation}
The gauge adopted here is therefore
\begin{equation}
n^\mu h_{\mu\nu}=0,
\qquad
\text{with no additional condition }h=0.
\label{eq:tracefulORG}
\end{equation}
We refer to this choice as the \emph{traceful radiation gauge}. It preserves
the algebraic simplifications associated with radiation gauges while allowing
the matter-sourced trace to be determined by the field equations. This
statement does not identify the trace as an additional propagating degree of
freedom.

The coordinate gauge condition alone does not fix the perturbation of
the null tetrad.  We adopt the tetrad gauge used in
Ref.~\cite{Li:2026rkf},
\begin{align}
l^{\mu(1)}=\frac12 h_{ll}^{(1)}n^\mu,
\qquad n^{\mu(1)}=0,\qquad
m^{\mu(1)}=h_{lm}^{(1)}n^\mu
-\frac12h_{m\bar m}^{(1)}m^\mu
-\frac12h_{mm}^{(1)}\bar m^\mu .
\label{eq:tetrad_gauge}
\end{align}
This choice follows from tetrad orthogonality together with the three
complex tetrad-rotation freedoms.  In particular, it implies
$\Delta^{(1)}=0$ and $\nu^{(1)}=0$ in the traceful outgoing radiation
gauge.  These conditions are essential for the reduced trace transport
equations.

After imposing Eq.~(\ref{eq:RG_components}), the remaining nonvanishing tetrad components are

\begin{equation}
h_{ll},
\qquad
h_{lm},
\qquad
h_{l\bar m},
\qquad
h_{mm},
\qquad
h_{m\bar m},
\qquad
h_{\bar m\bar m}.
\label{eq:surviving_components}
\end{equation}
The reality of the metric perturbation implies
\begin{equation}
h_{l\bar m}=
\overline{h_{lm}},
\qquad
h_{\bar m\bar m}=
\overline{h_{mm}}, \qquad \mbox{while $h_{ll}$ and $h_{m\bar m}$ are real.}
\end{equation}
Using the Newman--Penrose decomposition of the metric,
\begin{equation}
g_{\mu\nu}=
-l_\mu n_\nu
-n_\mu l_\nu
+m_\mu\bar m_\nu
+\bar m_\mu m_\nu,\quad \mbox{the trace of the perturbation becomes} \quad h
=
-2h_{ln}
+
2h_{m\bar m}.
\end{equation}
Since Eq.~(\ref{eq:traceful_RG}) implies $h_{nl}=0$ and the metric perturbation is symmetric, $h_{ln}=h_{nl}=0$. Consequently,

\begin{equation}
h
=
2h_{m\bar m}.
\label{eq:trace_component}
\end{equation}

Equation~(\ref{eq:trace_component}) shows that the radiation-gauge condition alone does not enforce tracelessness. Rather, the entire trace content of the perturbation is encoded in the single tetrad component $h_{m\bar m}$. In the conventional trace-free CCK framework this component is removed by imposing $h=0$, whereas in the present formulation it remains dynamical and participates in the reconstruction procedure.

The traceful radiation gauge of Li and Yunes~\cite{Li:2026rkf}
provides the starting point for the reconstruction programme adopted
here. It preserves the simplifications associated with the principal
null structure of Kerr while retaining sufficient generality to track
the evolution of the metric trace. As shown in the next section, the
reconstruction separates into a curvature sector governed by perturbed
Weyl scalars and a trace sector governed by the Ricci scalar
$\Phi_{22}^{(1)}$.

The use of radiation gauges is closely tied to the algebraically special character of the Kerr spacetime. Because Kerr is a Petrov type-D vacuum solution with only one non-vanishing background Weyl scalar, $\Psi_2$, gravitational perturbations may be described in terms of the extreme Weyl scalars $\Psi_0$ and $\Psi_4$, which satisfy the Teukolsky master equation \cite{Teukolsky:2014vca,Teukolsky:1973ha}. The CCK--Wald formalism reconstructs metric perturbations from these curvature variables in radiation gauges \cite{Chrzanowski:1975wv,Kegeles:1979an,Wald:1978vm}. The traceful radiation gauge adopted here extends this framework by retaining the trace sector and thereby provides the foundation for the Kerr reconstruction hierarchy developed in the following section.

\section{Kerr Reconstruction Hierarchy}\label{IV}

We now describe the hierarchical structure of the traceful Kerr reconstruction. The purpose of this hierarchy is to identify the sequence of curvature, connection, and metric variables required to reconstruct a metric perturbation on a Kerr background.

The construction is formulated in the Newman--Penrose formalism and exploits the Petrov type-D character of Kerr. The only nonvanishing background Weyl scalar is

\begin{equation}
\Psi_2^{(0)}=
-\frac{M}{(r-i a\cos\theta)^3},\qquad \mbox{while} \qquad
\Psi_0^{(0)}
=
\Psi_1^{(0)}
=
\Psi_3^{(0)}
=
\Psi_4^{(0)}
=
0.
\end{equation}
The radiative degrees of freedom are encoded in the first-order perturbations $\Psi_0^{(1)}$ and $\Psi_4^{(1)}$, which satisfy the Teukolsky master equation \cite{Teukolsky:2014vca,Teukolsky:1973ha}. In this work we focus on the spin$-2$ sector described by $\Psi_4^{(1)}$ and use it as the starting point of the reconstruction hierarchy.

The reconstruction is not an algebraic inversion of the metric. Rather, it is obtained by integrating a hierarchy of transport equations derived from the linearized Newman--Penrose system. The relevant equations include Bianchi identities, which govern the propagation of perturbed Weyl scalars and spin coefficients, together with Ricci identities, tetrad-connection equations, and linearized Einstein equations, which relate spin-coefficient perturbations to metric components. Therefore, the arrows in the schematic chains below denote transport relations obtained from the linearized Newman--Penrose equations rather than direct algebraic equalities.
\subsection*{Origin of the Newman--Penrose hierarchy}

The reconstruction hierarchy follows from the triangular structure of the linearized Newman--Penrose equations on a Petrov type-D background. Since the Kerr tetrad is aligned with the repeated principal null directions, the background quantities
\[
\kappa=\sigma=\lambda=\nu=0,
\qquad
\Psi_0=\Psi_1=\Psi_3=\Psi_4=0,
\]
vanish identically. As a result, their first-order perturbations satisfy simplified transport equations obtained by linearizing the Newman--Penrose Bianchi and Ricci identities.

The first class of equations is provided by the linearized Bianchi identities. These relate perturbations of the Weyl scalars to perturbations of the spin coefficients. In schematic form, the radiative and intermediate curvature sectors take the form
\begin{equation}
\Psi_4^{(1)}
\longrightarrow
\lambda^{(1)},
\qquad
\Psi_3^{(1)}
\longrightarrow
\pi^{(1)}.
\label{eq:NP_Bianchi_origin}
\end{equation}
Here the arrows denote first-order transport relations along the principal null directions, with coefficients determined by the Kerr background spin coefficients.

The second class of equations consists of the linearized Ricci identities and tetrad-connection relations. These equations relate the perturbed spin coefficients to tetrad components of the metric perturbation. Schematically,
\begin{equation}
\lambda^{(1)}
\longrightarrow
h_{\bar m\bar m}^{(1)},
\qquad
\pi^{(1)}
\longrightarrow
h_{l\bar m}^{(1)}.
\label{eq:NP_metric_origin}
\end{equation}

The trace sector arises from the Ricci identity governing the convergence coefficient $\mu$. In the traceful formulation this equation contains the Ricci scalar $\Phi_{22}^{(1)}$ as a source,
\begin{equation}
\Phi_{22}^{(1)}
\longrightarrow
\mu^{(1)}
\longrightarrow
h_{m\bar m}^{(1)}
\longrightarrow
h^{(1)}.
\label{eq:NP_trace_origin}
\end{equation}
This chain is absent in the conventional trace-free CCK construction, where $h^{(1)}=0$ is imposed and the component $h_{m\bar m}^{(1)}$ is removed.


Thus the full hierarchy is generated by a sequence of linearized Newman--Penrose transport equations: Bianchi identities determine the curvature-to-connection part of the reconstruction, while Ricci identities and connection equations determine the connection-to-metric part. The purpose of the following subsections is to organize these relations into curvature and trace sectors adapted to the traceful radiation gauge.

We impose the same radiation-gauge condition as in Sec.~\ref{III},
\begin{equation}
n^\mu h_{\mu\nu}^{(1)}=0, \quad \mbox{so that} \quad h_{nn}^{(1)}
=
h_{nl}^{(1)}
=
h_{nm}^{(1)}
=
h_{n\bar m}^{(1)}
=
0.
\label{eq:ORG_condition_hierarchy}
\end{equation}
The remaining tetrad components are
\begin{equation}
h_{ll}^{(1)},
\qquad
h_{lm}^{(1)},
\qquad
h_{l\bar m}^{(1)},
\qquad
h_{mm}^{(1)},
\qquad
h_{m\bar m}^{(1)},
\qquad
h_{\bar m\bar m}^{(1)}.
\end{equation}
The hierarchy should therefore be viewed as an
organizational structure for solving the coupled
Newman--Penrose system rather than as a sequence
of independent equations.
\subsection*{Curvature Sector}

The highest level of the curvature hierarchy is governed by the radiative Weyl scalar $\Psi_4^{(1)}$. Since

\begin{equation}
\Psi_4^{(0)}=0,
\qquad
\lambda^{(0)}=0,
\end{equation}
the relevant linearized Bianchi identity determines the perturbed shear coefficient $\lambda^{(1)}$ from $\Psi_4^{(1)}$, up to homogeneous solutions and boundary data. The subsequent Newman--Penrose connection equations then determine the metric component $h_{\bar m\bar m}^{(1)}$. Schematically,

\begin{equation}
\Psi_4^{(1)}
\longrightarrow
\lambda^{(1)}
\longrightarrow
h_{\bar m\bar m}^{(1)}.
\label{eq:spin2hierarchy}
\end{equation}
The first arrow represents a transport relation arising
from the linearized Bianchi identities, while the second
represents the determination of the corresponding metric
component through the linearized Newman--Penrose field
equations.

It is useful to distinguish the completion sector from
the trace sector introduced below. The completion problem
is associated with perturbations that are not determined
by the radiative Weyl scalars alone, most notably
variations of the background Kerr parameters,

\begin{equation}
\delta M,
\qquad
\delta J .
\end{equation}

These perturbations are encoded in the Coulombic part of
the gravitational field and are naturally associated with
the Weyl scalar $\Psi_2^{(1)}$.
By contrast, the trace sector considered in the present
work is generated by Ricci curvature produced by matter
sources. Its entry point is the Ricci scalar
$\Phi_{22}^{(1)}$, which drives the transport equation
for $\mu^{(1)}$ and subsequently determines
$h_{m\bar m}^{(1)}$ and the metric trace. Thus, the trace hierarchy
\begin{equation}\label{123}
\Phi_{22}^{(1)}
\longrightarrow
\mu^{(1)}
\longrightarrow
h_{m\bar m}^{(1)}
\longrightarrow
h^{(1)},
\end{equation}
does not belong to the completion sector in the conventional
CCK sense. Rather, it constitutes an independent Ricci-driven
reconstruction channel associated with matter sources and the
metric trace. This sector is absent in the standard vacuum,
trace-free reconstruction framework.

The two sectors may coexist for sourced perturbations:
the completion sector accounts for homogeneous
Coulombic perturbations such as $\delta M$ and
$\delta J$, while the trace sector accounts for the
particular Ricci-sourced response associated with
$T_{\mu\nu}\neq0$. Keeping these sectors conceptually
separate is essential for identifying which part of the
metric perturbation arises from completion data and
which part is driven directly by matter sources.

The next level is controlled by $\Psi_3^{(1)}$. Since

\begin{equation}
\Psi_3^{(0)}=0,
\end{equation}
the corresponding linearized Bianchi identity determines the perturbed spin coefficient $\pi^{(1)}$. The Newman--Penrose connection equations then determine the mixed metric component $h_{l\bar m}^{(1)}$. Thus,

\begin{equation}
\Psi_3^{(1)}
\longrightarrow
\pi^{(1)}
\longrightarrow
h_{l\bar m}^{(1)}.
\label{eq:spin1hierarchy}
\end{equation}

The final curvature level is associated with $\Psi_2^{(1)}$. Unlike $\Psi_3^{(1)}$ and $\Psi_4^{(1)}$, the scalar $\Psi_2$ is nonzero on the Kerr background and describes the Coulombic part of the gravitational field. This sector is closely related to the completion problem of metric reconstruction, since perturbations corresponding to variations of the background mass and angular momentum are not completely determined by the radiative degrees of freedom alone \cite{Wald:1978vm,Ori:2002uv,Keidl:2010pm,Merlin:2016boc}. Within the present hierarchy, the transport equations associated with $\Psi_2^{(1)}$ contribute to the determination of the remaining longitudinal metric component,

\begin{equation}
\Psi_2^{(1)}
\longrightarrow
h_{ll}^{(1)}.
\label{eq:coulombhierarchy}
\end{equation}

The curvature sector is therefore summarized as

\begin{equation}
\Psi_4^{(1)}
\longrightarrow
\lambda^{(1)}
\longrightarrow
h_{\bar m\bar m}^{(1)},
\qquad
\Psi_3^{(1)}
\longrightarrow
\pi^{(1)}
\longrightarrow
h_{l\bar m}^{(1)},
\qquad
\Psi_2^{(1)}
\longrightarrow
h_{ll}^{(1)}.
\label{eq:curvature_hierarchy}
\end{equation}
This sector reproduces, at the level of Newman--Penrose transport equations, the logical structure underlying radiation-gauge metric reconstruction on Kerr.

\subsection*{Trace Sector}

The principal extension introduced in the traceful formulation is the retention of the metric trace. In the conventional trace-free CCK reconstruction one imposes

\begin{equation}
h^{(1)}=0,\quad \mbox{which, in the present gauge, is equivalent to setting}\quad h_{m\bar m}^{(1)}=0.
\end{equation}

Here this restriction is not imposed.
The trace sector is generated through the perturbation
of the spin coefficient $\mu$. To write the corresponding
transport equations compactly, we introduce the weighted
Newman--Penrose transport operator

\begin{equation}
\Delta_{[p,q,r,s]} X
\equiv
\left(
\Delta+p\,\mu+q\,\bar\mu+r\,\gamma+s\,\bar\gamma
\right)X,
\qquad \Delta\equiv n^\mu\nabla_\mu .
\label{eq:weighted_Delta_operator}
\end{equation}

The ordering of the weights in this definition follows the convention of
Ref.~\cite{Li:2026rkf}.  The reduced equation is obtained directly from the
linearized Ricci identity,
\begin{equation}
\Delta_{[2,0,1,1]}
\mu^{(1)}
+\mu^{(0)}\left(\gamma^{(1)}+\bar\gamma^{(1)}\right)
+\Phi_{22}^{(1)}=0 .
\label{eq:mutransport}
\end{equation}
Here $\Phi_{22}^{(1)}=R_{\mu\nu}^{(1)}n^\mu n^\nu/2$.
The commutation relations in the tetrad gauge
\eqref{eq:tetrad_gauge} yield
\begin{equation}
\mu^{(1)}=\frac12\Delta_{[-1,1,0,0]}
h_{m\bar m}^{(1)},
\qquad
\gamma^{(1)}=-\frac14(\mu-\bar\mu)h_{m\bar m}^{(1)}.
\label{eq:mu_gamma_metric}
\end{equation}
Because $h_{m\bar m}^{(1)}$ is real,
$\gamma^{(1)}+\bar\gamma^{(1)}=0$.  Equation
\eqref{eq:mutransport} therefore reduces to
\begin{equation}
\Delta_{[2,0,1,1]}\mu^{(1)}=-\Phi_{22}^{(1)}.
\label{eq:mutransport_reduced}
\end{equation}
This reduction also uses $\lambda^{(0)}=\nu^{(0)}=0$ for the type-D
background and $\Delta^{(1)}=\nu^{(1)}=0$ in the tetrad gauge above
\cite{Li:2026rkf}.  We denote the Kerr radial function by
$\Delta_K=r^2-2Mr+a^2$ to distinguish it from the directional
derivative $\Delta$.

For stationary axisymmetric fields, $\partial_t X=\partial_\phi X=0$, and therefore
\begin{equation}
n^\mu\nabla_\mu X=
-\frac{\Delta_K}{2\Sigma}\,
\partial_r X .
\label{eq:n_direction_stationary}
\end{equation}
Thus the operator $\Delta_{[p,q,r,s]}$ becomes a first-order radial transport operator after angular-mode projection.  In vacuum, the Ricci-driven trace equations become homogeneous.
With boundary conditions and residual gauge choices that remove the
homogeneous trace sector, the reconstruction reduces to the usual
vacuum radiation-gauge formulation.

Once $\mu^{(1)}$ is determined, the trace component of the metric perturbation follows from

\begin{equation}
\Delta_{[-1,1,0,0]}h_{m\bar m}^{(1)}=2\mu^{(1)}.
\label{eq:trace_transport}
\end{equation}

Using Eq.~(\ref{eq:trace_component}) at first order gives

\begin{equation}
h^{(1)}=
2h_{m\bar m}^{(1)}.
\end{equation}

Therefore the trace hierarchy is

\begin{equation}
\Phi_{22}^{(1)}
\longrightarrow
\mu^{(1)}
\longrightarrow
h_{m\bar m}^{(1)}
\longrightarrow
h^{(1)}.
\label{eq:tracehierarchy}
\end{equation}

The reconstruction procedure thus separates naturally into a curvature
sector governed by Weyl scalars and a trace sector governed by
$\Phi_{22}^{(1)}$.  This Ricci-driven trace sector is the essential
ingredient introduced by Li and Yunes~\cite{Li:2026rkf}; here we apply
their framework to the closed shell model developed below.

\section{Application I: A Closed Schwarzschild Shell and Its Slowly Rotating Sector}
\label{Vclosed}

We first present a closed static-shell calculation in which the source,
differential equations, boundary data, junction conditions, and
physical charges are fixed.  We then close the leading axial dipole by
specifying a rigid shell angular velocity and deriving the associated
surface current and angular momentum.  The construction follows the
static-shell test of the traceful reconstruction method in
Ref.~\cite{Li:2026rkf}; the rotating shell developed below supplies
the additional slowly rotating sector.

\subsection{Two-parameter perturbative ordering}

The shell problem contains two independent small quantities. We
introduce
\begin{equation}
\epsilon_M\equiv\frac{\delta M^{(1,0)}}{M}
=\frac{\sqrt{f(R)}M_{\mathrm{s}}}{M},
\qquad
\epsilon_\Omega\equiv
\frac{R\Omega_{\mathrm{s}}}{\sqrt{f(R)}} ,
\label{eq:two_parameters}
\end{equation}
without assuming that their numerical magnitudes are equal. The metric
and surface stress tensor are organized as bivariate expansions,
\begin{align}
g_{\mu\nu}(\epsilon_M,\epsilon_\Omega)
&=g_{\mu\nu}^{(0,0)}
+\epsilon_M h_{\mu\nu}^{(1,0)}
+\epsilon_M\epsilon_\Omega h_{\mu\nu}^{(1,1)}
+\epsilon_M\epsilon_\Omega^2h_{\mu\nu}^{(1,2)}
+O(\epsilon_M^2,\epsilon_M\epsilon_\Omega^3),
\label{eq:bivariate_metric}\\
S_{AB}(\epsilon_M,\epsilon_\Omega)
&=\epsilon_M S_{AB}^{(1,0)}
+\epsilon_M\epsilon_\Omega S_{AB}^{(1,1)}
+\epsilon_M\epsilon_\Omega^2S_{AB}^{(1,2)}
+O(\epsilon_M^2,\epsilon_M\epsilon_\Omega^3).
\label{eq:bivariate_surface}
\end{align}
The static spherical perturbation belongs to the $(1,0)$ sector, the
leading axial response belongs to the mixed $(1,1)$ sector, and the
quadratic even-parity response belongs to $(1,2)$. There is no
$(0,1)$ contribution because an angular velocity
produces no shell stress-energy when the shell mass vanishes. Thus the
calculation remains linear in the shell strength but is carried
through quadratic order in the independent rotation parameter,
without requiring $\epsilon_M$ and $\epsilon_\Omega$ to have the same
magnitude.
In Secs.~\ref{II}--\ref{IV}, the superscript $(1)$ retains its usual
meaning of a generic linear perturbation. In the explicit shell
calculation, a static quantity carrying this superscript belongs to
the $(1,0)$ sector, whereas an axial quantity carrying it belongs to
the $(1,1)$ sector.

A systematic extension can be organized as
\begin{equation}
g_{\mu\nu}=\sum_{p,q\geq0}\epsilon_M^p\epsilon_\Omega^q
g_{\mu\nu}^{(p,q)} .
\label{eq:general_bivariate}
\end{equation}
At each bidegree $(p,q)$ the linearized Einstein operator acts on the
new coefficient $g_{\mu\nu}^{(p,q)}$, while its source contains the
matter coefficient at that order and combinations already determined
at lower bidegrees. The sectors $(2,0)$, $(1,2)$, and $(2,1)$ describe,
respectively, static nonlinear backreaction, rotational deformation
and quadrupolar structure, and nonlinear mass--rotation coupling. This
paper solves the $(1,0)$, $(1,1)$, and $(1,2)$ shell sectors and
obtains the exact static mass through $(2,0)$; the complete
$h_{\mu\nu}^{(2,0)}$ field and the $(2,1)$ sector are not reconstructed.

To make the order-by-order prescription explicit, write the Einstein
tensor as
\begin{equation}
G_{\mu\nu}[g^{(0,0)}+h]
=G_{\mu\nu}^{(0,0)}
+G_{\mu\nu}^{[1]}[h]
+G_{\mu\nu}^{[2]}[h,h]+\cdots ,
\label{eq:einstein_expansion}
\end{equation}
where $G_{\mu\nu}^{[2]}$ is defined as a symmetric bilinear operator.
The first equations in the bivariate hierarchy are then
\begin{align}
G_{\mu\nu}^{[1]}[h^{(1,0)}]
&=8\pi T_{\mu\nu}^{(1,0)},\label{eq:order10}\\
G_{\mu\nu}^{[1]}[h^{(1,1)}]
&=8\pi T_{\mu\nu}^{(1,1)},\label{eq:order11}\\
G_{\mu\nu}^{[1]}[h^{(2,0)}]
&=8\pi T_{\mu\nu}^{(2,0)}
-G_{\mu\nu}^{[2]}[h^{(1,0)},h^{(1,0)}],
\label{eq:order20}\\
G_{\mu\nu}^{[1]}[h^{(1,2)}]
&=8\pi T_{\mu\nu}^{(1,2)},\label{eq:order12}\\
G_{\mu\nu}^{[1]}[h^{(2,1)}]
&=8\pi T_{\mu\nu}^{(2,1)}
-2G_{\mu\nu}^{[2]}[h^{(1,0)},h^{(1,1)}].
\label{eq:order21}
\end{align}
Thus the same linearized operator is solved at every bidegree, with
lower-order gravitational fields entering as known effective sources.
At a given bidegree, the corresponding effective Ricci tensor can be
projected on the background Newman--Penrose tetrad to generate source
terms for a linear transport problem. This supplies systematic
Einstein-equation bookkeeping, but it is not by itself a proof of a
complete nonlinear radiation-gauge reconstruction. At higher orders
one must additionally include tetrad and gauge corrections and verify
that the chosen radiation gauge remains admissible. Nor does the
hierarchy remove the need to specify the material coefficients
$T_{\mu\nu}^{(p,q)}$.

\subsection{Conserved static shell}

Let
\begin{equation}
 f(r)=1-\frac{2M}{r}
\end{equation}
and place a spherical shell of areal radius $R>2M$ and rest mass
$M_{\mathrm {s}}\ll M$ around the black hole.  The spacetime on each side of
the shell is Schwarzschild,
\begin{equation}
 ds_\pm^2=-f_\pm(r_\pm)dt_\pm^2+f_\pm^{-1}(r_\pm)dr_\pm^2
 +r_\pm^2d\Omega^2,
 \qquad f_\pm(r)=1-\frac{2M_\pm}{r},
\label{eq:shell_regions}
\end{equation}
with $M_-=M$ and $M_+=M+\delta M$.  The Israel junction
conditions~\cite{Israel:1966rt}
give the surface tensor
\begin{equation}
 S_{AB}=\sigma u_Au_B+p(q_{AB}+u_Au_B),
\label{eq:surface_tensor}
\end{equation}
where, writing $\alpha_\pm=\sqrt{f_\pm(R)}$,
\begin{align}
 \sigma=\frac{\alpha_--\alpha_+}{4\pi R},\qquad 
 p=\frac{1}{8\pi R^2}
 \left(\frac{R-M_+}{\alpha_+}-\frac{R-M_-}{\alpha_-}\right).
\label{eq:sigma_pressure_exact}
\end{align}
Defining the shell rest mass by
$M_{\mathrm{s}}=4\pi R^2\sigma$, the first of
Eqs.~\eqref{eq:sigma_pressure_exact} can be solved without expanding:
\begin{equation}
\sqrt{f_+(R)}
=\sqrt{f(R)}-\frac{M_{\mathrm{s}}}{R}.
\label{eq:exact_alpha_relation}
\end{equation}
The chosen branch requires
$0<M_{\mathrm{s}}<R\sqrt{f(R)}$.
Squaring this relation gives the exact mass relation for the static
thin shell,
\begin{equation}
{\delta M
=\sqrt{f(R)}\,M_{\mathrm{s}}
-\frac{M_{\mathrm{s}}^2}{2R}}.
\label{eq:exact_mass_relation}
\end{equation}
Consequently, the first two static orders are
\begin{equation}
\delta M^{(1,0)}=\sqrt{f(R)}\,M_{\mathrm{s}},
\qquad
\delta M^{(2,0)}=-\frac{M_{\mathrm{s}}^2}{2R}.
\label{eq:mass_orders}
\end{equation}
Equivalently, in terms of the parameter in
Eq.~\eqref{eq:two_parameters},
\begin{equation}
\frac{\delta M}{M}
=\epsilon_M-\frac{M}{2Rf(R)}\epsilon_M^2 .
\label{eq:mass_lambda_expansion}
\end{equation}
The second term is the leading binding-energy correction and provides
an explicit higher-order result within the static sector.

The pressure is not optional: it supplies the force required to keep
the shell static and makes the distributional stress tensor conserved.
At the first static order,
\begin{equation}
 \delta M=\sqrt{f(R)}\,M_{\mathrm {s}},
 \qquad
 \sigma=\frac{\delta M}{4\pi R^2\sqrt{f(R)}},
 \qquad
 p=\frac{M\delta M}{8\pi R^3 f^{3/2}(R)}.
\label{eq:linear_shell_data}
\end{equation}
Here and in the first-order Newman--Penrose calculation below,
$\delta M$ on the right-hand side denotes its leading value
$\sqrt{f(R)}M_{\mathrm{s}}$. The exact pressure may also be expressed
in terms of $M_{\mathrm{s}}$ as
\begin{equation}
p=
\frac{M_{\mathrm{s}}}{8\pi R^2}
\frac{M/R+\sqrt{f(R)}\,M_{\mathrm{s}}/(2R)}
{\sqrt{f(R)}
\left[\sqrt{f(R)}-M_{\mathrm{s}}/R\right]} .
\label{eq:exact_pressure_ms}
\end{equation}
Its expansion begins as
\begin{equation}
p=
\frac{M M_{\mathrm{s}}}{8\pi R^3 f(R)}
+\frac{M_{\mathrm{s}}^2}{16\pi R^3 f^{3/2}(R)}
+O(M_{\mathrm{s}}^3/R^4).
\label{eq:pressure_second_order}
\end{equation}
At the perturbative order considered here, the ratio of the tangential
pressure to the surface energy density is
\begin{equation}
\frac{p}{\sigma}
=
\frac{M}{2R f(R)}
=
\frac{M}{2(R-2M)}.
\label{eq:pressure_density_ratio}
\end{equation}
Since $\sigma>0$ and $p>0$ for $R>2M$, the weak energy condition is
satisfied. The dominant energy condition additionally requires
$p\leq\sigma$, which gives
\begin{equation}
R\geq\frac{5M}{2}.
\label{eq:dominant_energy_condition}
\end{equation}
Consequently, shells with $2M<R<5M/2$ remain mathematically admissible
thin-shell solutions but violate the dominant energy condition because
their tangential pressure exceeds their surface energy density. When
a shell satisfying the dominant energy condition is required, the
physical domain of the present first-order model should therefore be
restricted to $R\geq5M/2$. 
The four-dimensional source is
\begin{equation}
 T^{\mu\nu}=\left[\sigma u^\mu u^\nu
 +p\left(e_{\hat\theta}^{\mu}e_{\hat\theta}^{\nu}
 +e_{\hat\phi}^{\mu}e_{\hat\phi}^{\nu}\right)\right]\delta(\ell),
 \qquad
 \ell=\frac{r-R}{\sqrt{f(R)}}.
\label{eq:conserved_shell_source}
\end{equation}
Equations \eqref{eq:sigma_pressure_exact} and
\eqref{eq:conserved_shell_source} imply
$\nabla_\mu T^{\mu\nu}=0$ in the distributional sense.

On the Schwarzschild background we use
\begin{equation}
 l^\mu=\left(f^{-1},1,0,0\right),\qquad
 n^\mu=\frac12\left(1,-f,0,0\right),\qquad
 m^\mu=\frac{1}{\sqrt{2}r}(0,0,1,i\csc\theta).
\label{eq:schwarzschild_tetrad}
\end{equation}
The nonzero spin coefficients required below are
\begin{equation}
 \rho=-\frac1r,\qquad
 \mu=-\frac{f}{2r},\qquad
 \gamma=\frac{M}{2r^2},\qquad
 \alpha=-\beta=-\frac{\cot\theta}{2\sqrt2 r}.
\end{equation}
Following the static-shell calculation of Li and
Yunes~\cite{Li:2026rkf}, the parameter
$\epsilon_M=\delta M^{(1,0)}/M$ introduced
in Eq.~\eqref{eq:two_parameters} controls the static expansion. The
nonzero coefficients of the first-order NP source
projections are
\begin{align}
\Phi_{00}^{(1)}
&=
\frac{M}{R^2f(R)}\delta(r-R),
&
\Phi_{22}^{(1)}
&=
\frac{Mf(R)}{4R^2}\delta(r-R),
\nonumber\\
\Phi_{11}^{(1)}
&=
\frac{M(R-M)}{4R^3f(R)}\delta(r-R),
&
\Lambda_{\mathrm{NP}}^{(1)}
&=
\frac{M(R-3M)}{12R^3f(R)}\delta(r-R).
\label{eq:shell_ricci_scalars}
\end{align}
Here, $\Lambda_{\mathrm{NP}}^{(1)}=R^{(1)}/24$ and should not be confused with the cosmological constant. The appearance of all four projections explains why $\Phi_{22}^{(1)}$ alone is insufficient to close the full metric perturbation problem.

Spherical symmetry sets the spin-2 and spin-1 amplitudes to zero,
\begin{equation}
 \Psi_4^{(1)}=\lambda^{(1)}=h_{\bar m\bar m}^{(1)}=0,
 \qquad
 \Psi_3^{(1)}=\pi^{(1)}=h_{l\bar m}^{(1)}=0.
\end{equation}
The trace equations reduce to
\begin{align}
 \left(\frac{f}{2}\frac{d}{dr}+\frac{r-3M}{r^2}\right)
 \mu^{(1)}&=\Phi_{22}^{(1)},\label{eq:mu_shell_ode}\\
 \frac{f}{2}\frac{d h_{m\bar m}^{(1)}}{dr}&=-2\mu^{(1)}.
\label{eq:hmm_shell_ode}
\end{align}
Eliminating $\mu^{(1)}$ gives the distributional equation
\begin{equation}
 \frac{d}{dr}\left(r^2\frac{d h_{m\bar m}^{(1)}}{dr}\right)
 =-\frac{2M}{f(R)}\delta(r-R).
\label{eq:hmm_decoupled}
\end{equation}
We require a regular interior solution, an asymptotically Schwarzschild
exterior, continuity of $h_{m\bar m}^{(1)}$, and the jump condition
obtained by integrating Eq.~\eqref{eq:hmm_decoupled}.  These conditions
fix all homogeneous constants and give
\begin{equation}
 h_{m\bar m}^{(1)}=\frac{2M}{Rf(R)}
 \left(1-\frac{R}{r}\right)\Theta(R-r),
 \qquad
 \mu^{(1)}=-\frac{Mf(r)}{2f(R)r^2}\Theta(R-r).
\label{eq:trace_solution_shell}
\end{equation}
Thus
\begin{equation}
 h^{(1)}=2h_{m\bar m}^{(1)}
 =\frac{4M}{Rf(R)}\left(1-\frac{R}{r}\right)\Theta(R-r).
\label{eq:metric_trace_shell}
\end{equation}

The remaining spin-zero completion equation contains
$\Psi_2^{(1)}$, $\Lambda_{\mathrm { NP}}^{(1)}$, and $\mu^{(1)}$.
Solving it with a mass shift $\delta M$ outside the shell and imposing
continuity at $r=R$ gives
\begin{align}
 h_{ll}^{+(1)}=\frac{8M}{r f^2(r)},\qquad 
 h_{ll}^{-(1)}=\frac{8M}{Rf^2(r)f(R)}
 \left(1-\frac{3M}{r}+\frac{MR}{r^2}\right).
\label{eq:hll_shell_solution}
\end{align}
The superscripts $+$ and $-$ denote $r>R$ and $2M<r<R$.
In the traceful outgoing radiation gauge,
\begin{align}
h_{tt}^{(1)}
&=
f(r)h_{tr}^{(1)}
=
f^2(r)h_{rr}^{(1)}
=
\frac{f^2(r)}{4}h_{ll}^{(1)},
\nonumber\\
h_{\phi\phi}^{(1)}
&=
\sin^2\theta\,h_{\theta\theta}^{(1)}
=
r^2\sin^2\theta\,h_{m\bar m}^{(1)}.
\label{eq:shell_coordinate_components}
\end{align}
Within the tetrad, gauge, and delta-function conventions adopted here,
these expressions reproduce the static-shell solution reported by Li
and Yunes~\cite{Li:2026rkf}.  The exterior $1/r$ coefficient gives
the definite observable
\begin{equation}
 {\delta M_{\mathrm {ADM}}=\sqrt{f(R)}\,M_{\mathrm {s}}}.
\label{eq:adm_mass_result}
\end{equation}

\subsection{Closed slowly rotating axial sector}

We now allow the shell to rotate rigidly with angular velocity
$\Omega_{\mathrm{s}}$ as measured relative to the asymptotic time
coordinate. Rotation is treated to first order in $\epsilon_\Omega$,
and the
retained axial response is the mixed $(1,1)$ sector. Changes in the
surface density, pressure, and shape
begin only at second rotational order and do
not enter the present calculation, as in the standard slow-rotation
expansion~\cite{Hartle:1967he}.  The shell four-velocity is
\begin{equation}
u^\mu=u^t(1,0,0,\Omega_{\mathrm{s}}),
\qquad
u^t=\frac{1}{\sqrt{f(R)}}+O(\Omega_{\mathrm{s}}^2).
\label{eq:rotating_shell_velocity}
\end{equation}
The surface tensor retains the perfect-fluid form
\begin{equation}
S^A{}_B=(\sigma+p)u^A u_B+p\,\delta^A{}_B,\quad \mbox{with $\sigma$ and $p$ given by Eq.~\eqref{eq:linear_shell_data}.
}
\label{eq:rotating_surface_tensor}
\end{equation}

Introduce the frame-dragging function in the standard slow-rotation
form~\cite{Hartle:1967he,Brill:1966rt}
\begin{equation}
\omega(r)=-\frac{h_{t\phi}^{(1)}}{r^2\sin^2\theta}.
\label{eq:omega_definition}
\end{equation}
The angular velocity of the shell relative to the local inertial
frames is $\Omega_{\mathrm{s}}-\omega(R)$. At the retained axial order,
the mixed angular-momentum surface density is therefore
\begin{equation}
S^t{}_{\phi}
=\frac{(\sigma+p)R^2\sin^2\theta}{f(R)}
\left[\Omega_{\mathrm{s}}-\omega(R)\right].
\label{eq:surface_angular_current}
\end{equation}
The corresponding conserved angular momentum is
\begin{align}
\delta J
&=\int T^t{}_{\phi}\sqrt{-g}\,
dr\,d\theta\,d\phi
=\frac{8\pi R^4}{3\sqrt{f(R)}}(\sigma+p)
\left[\Omega_{\mathrm{s}}-\omega(R)\right].
\label{eq:J_surface_integral}
\end{align}
It is convenient to define the relativistic shell moment of inertia
\begin{equation}
I_{\mathrm{s}}
\equiv
\frac{8\pi R^4}{3\sqrt{f(R)}}(\sigma+p).
\label{eq:shell_moment_inertia_definition}
\end{equation}
Using Eq.~\eqref{eq:linear_shell_data}, we obtain
\begin{equation}
{
I_{\mathrm{s}}
=\frac{\delta M R(2R-3M)}{3f^2(R)}
}.
\label{eq:shell_moment_inertia}
\end{equation}
This expression reduces to
$I_{\mathrm{s}}=2M_{\mathrm{s}}R^2/3$ in the weak-field limit, as
required for a thin spherical shell.

The distributional $t\phi$ component of the linearized Einstein
equations reduces to
\begin{equation}
\frac{d}{dr}\left(r^4\frac{d\omega}{dr}\right)
=-6\delta J\,\delta(r-R).
\label{eq:sourced_dragging_equation}
\end{equation}
Away from the shell the right-hand side vanishes.  Asymptotic flatness
eliminates the constant exterior mode.  In the interior we impose that
the central black hole has no intrinsic angular momentum, which removes
the $r^{-3}$ homogeneous mode.  Continuity of the induced metric then
gives
\begin{equation}
\omega(r)=
\begin{cases}
2\delta J/R^3, & 2M<r<R,\\[2mm]
2\delta J/r^3, & r>R.
\end{cases}
\label{eq:dragging_solution_closed}
\end{equation}
Integration of Eq.~\eqref{eq:sourced_dragging_equation} across the shell
gives the junction condition
\begin{equation}
\left[r^4\omega'(r)\right]_{R^-}^{R^+}=-6\delta J, \quad \mbox{which is satisfied by Eq.~\eqref{eq:dragging_solution_closed}.}
\label{eq:dragging_jump_closed}
\end{equation}

In the traceful outgoing radiation gauge, the nonzero axial components
may be written as
\begin{equation}
h_{t\phi}^{(1)}=-r^2\sin^2\theta\,\omega(r),
\qquad
h_{r\phi}^{(1)}=\frac{h_{t\phi}^{(1)}}{f(r)}, \quad \mbox{and hence $n^\mu h_{\mu\phi}^{(1)}=0$. }
\label{eq:axial_org_components_closed}
\end{equation}
 Although $h_{t\phi}^{(1)}$ remains finite at $r=2M$, the component
$h_{r\phi}^{(1)}=h_{t\phi}^{(1)}/f(r)$ diverges there in Schwarzschild
coordinates. This divergence is a coordinate feature of the
radiation-gauge representation and does not by itself indicate a
curvature singularity. Indeed, introducing the horizon-penetrating
advanced coordinate
\begin{equation}
v=t+r_*,
\qquad
\frac{dr_*}{dr}=\frac{1}{f(r)}, \quad \mbox{gives} \quad 
dt=dv-\frac{dr}{f(r)}.
\end{equation}
Consequently, the apparently divergent radial cross term cancels,
\begin{equation}
h_{r\phi}^{(1)}
-\frac{h_{t\phi}^{(1)}}{f(r)}
=0.
\end{equation}
 Substitution into the
distributional $t\phi$ Einstein equation reproduces the surface current
in Eq.~\eqref{eq:surface_angular_current}.

Since $I_{\mathrm{s}}=O(\epsilon_M)$ and
$\omega(R)=O(\epsilon_M\epsilon_\Omega)$,
Eq.~\eqref{eq:J_surface_integral} gives, at the retained mixed
bidegree $(1,1)$,
\begin{equation}
\delta J=I_{\mathrm{s}}\Omega_{\mathrm{s}}
\left[1+O(\epsilon_M,\epsilon_\Omega^2)\right].
\label{eq:J_Omega_relation}
\end{equation}
We do not retain the formally resummed denominator generated by
substituting $\omega(R)=2\delta J/R^3$ into
Eq.~\eqref{eq:J_surface_integral}, because that denominator contains
terms of order $\epsilon_M^2\epsilon_\Omega$ and higher, beyond the
bidegrees
retained here.

The constant interior dragging and the $r^{-3}$ exterior behavior are
the standard first-order results for a slowly rotating spherical
shell~\cite{Brill:1966rt}.  The coefficient of the exterior
$r^{-3}$ solution gives
\begin{equation}
{J_{\mathrm{Komar}}=\delta J},
\label{eq:komar_angular_momentum}
\end{equation}
while the interior inertial frames rotate rigidly at
\begin{equation}
{
\Omega_{\mathrm{drag}}
=\frac{2\delta J}{R^3}
=\frac{2\delta M(2R-3M)}{3R^2f^2(R)}
\Omega_{\mathrm{s}}
}
+O(\delta M^2\Omega_{\mathrm{s}}).
\label{eq:dragging_rate_closed}
\end{equation}
Thus the dragging fraction is
\begin{equation}
{
\frac{\Omega_{\mathrm{drag}}}{\Omega_{\mathrm{s}}}
=\frac{2\delta M(2R-3M)}{3R^2f^2(R)}
}
+O(\delta M^2).
\label{eq:dragging_fraction}
\end{equation}
The perturbative validity of these results requires more than
$M_{\mathrm{s}}\ll M$ when the shell approaches the horizon. Because
the coefficients of the static and axial perturbations contain inverse
powers of $f(R)$, a sufficient set of conditions for the present
first-order treatment is
\begin{equation}
\frac{\delta M}{R f(R)}\ll1,
\qquad
\frac{\delta M(2R-3M)}{R^2f^2(R)}\ll1,
\label{eq:near_horizon_mass_conditions}
\end{equation}
together with the slow-rotation condition
\begin{equation}
\frac{R\left|\Omega_{\mathrm{s}}-\omega(R)\right|}
{\sqrt{f(R)}}\ll1.
\label{eq:local_rotation_condition}
\end{equation}
The first condition controls the static metric perturbation, while the
second controls the rotational response and the expansion leading to
Eq.~\eqref{eq:J_Omega_relation}. The last condition requires the
physical tangential velocity of the shell, measured relative to the
local inertial frames, to remain small. Consequently, although the
shell is formally located at $R>2M$, the present first-order expansion
is not uniform in the limit $R\rightarrow2M$ and should not be used
arbitrarily close to the horizon. 
Together, Eqs.~\eqref{eq:adm_mass_result},
\eqref{eq:shell_moment_inertia}, \eqref{eq:J_Omega_relation}, and
\eqref{eq:dragging_fraction} provide definite mass, rotational, and
frame-dragging observables for the closed shell model.

The next rotational sector, $(1,2)$, is qualitatively different from
the axial calculation above. At that order the shell generally
develops an even-parity deformation, and its surface density and
stresses acquire rotational corrections. Conservation and the
junction conditions alone do not determine all of these material
responses. An equation of state or an equivalent constitutive law,
together with a prescription for the perturbed shell shape, is
therefore required to determine $T_{\mu\nu}^{(1,2)}$ and close
Eq.~\eqref{eq:order12}. We do not impose such additional material
physics in the leading axial calculation.  Instead, the following
subsection states the additional assumptions separately and uses them
to turn the $(1,2)$ sector into a closed, testable extension of the
benchmark.

\subsection{Barotropic closure at quadratic order in rotation}
\label{subsec:barotropic_shell}

We now make the additional material assumption needed to formulate the
next rotational sector. Following the standard Hartle construction
\cite{Hartle:1967he} and its thin-shell implementation
\cite{Uchikata:2016}, the shell is displaced according to
\begin{equation}
r_{\Sigma}(\theta)
=R+\epsilon_\Omega^2
\left[\xi_0+\xi_2P_2(\cos\theta)\right]
+O(\epsilon_\Omega^4),
\label{eq:shell_displacement_second}
\end{equation}
and the even-parity surface variables take the form
\begin{align}
\sigma(\theta)
&=\sigma_0+\epsilon_\Omega^2
\left[\delta\sigma_0+\delta\sigma_2P_2(\cos\theta)\right]
+O(\epsilon_\Omega^4),\label{eq:sigma_second_rotation}\\
p(\theta)
&=p_0+\epsilon_\Omega^2
\left[\delta p_0+\delta p_2P_2(\cos\theta)\right]
+O(\epsilon_\Omega^4).
\label{eq:p_second_rotation}
\end{align}
Reflection symmetry and invariance under reversal of the rotation
direction exclude odd powers of $\epsilon_\Omega$ from these scalar
quantities. The quadratic rotational sector therefore contains only
the monopole and quadrupole harmonics.

For a barotropic shell, $p=p(\sigma)$, the material perturbations obey
\begin{equation}
\delta p_\ell=c_{\mathrm{s}}^2\delta\sigma_\ell,
\qquad \ell=0,2,
\qquad
c_{\mathrm{s}}^2\equiv
\left.\frac{dp}{d\sigma}\right|_{\sigma_0}.
\label{eq:shell_barotropic_closure}
\end{equation}
This is the constitutive information absent from the leading axial
calculation. As a definite closure, consider the one-parameter family
of leading shell backgrounds obtained by varying $R$ at fixed
$M_{\mathrm{s}}$. From
\begin{equation}
\sigma_0=\frac{M_{\mathrm{s}}}{4\pi R^2},
\qquad
p_0=\frac{M M_{\mathrm{s}}}
{8\pi R^2(R-2M)},
\label{eq:background_shell_eos_family}
\end{equation}
we obtain
\begin{equation}
c_{\mathrm{s}}^2
=\frac{dp_0/dR}{d\sigma_0/dR}
=\frac{M(3R-4M)}{4(R-2M)^2}.
\label{eq:shell_sound_speed}
\end{equation}
Thus this particular barotropic sequence is causal,
$0\leq c_{\mathrm{s}}^2\leq1$, only when
\begin{equation}
\frac{R}{M}\geq
\frac{19+\sqrt{41}}{8}
\approx3.175 .
\label{eq:shell_causal_radius}
\end{equation}
This restriction is stronger than the leading dominant-energy
condition $R/M\geq5/2$.

The metric at the same rotational order may be written in the
Hartle form
\begin{align}
ds^2={}&-f(r)\left[1+2\epsilon_\Omega^2
\left(h_0+h_2P_2\right)\right]dt^2
+\frac{1}{f(r)}
\left[1+\frac{2\epsilon_\Omega^2
\left(m_0+m_2P_2\right)}{rf(r)}\right]dr^2
\nonumber\\
&+r^2\left[1+2\epsilon_\Omega^2 k_2P_2\right]
\left\{d\theta^2+\sin^2\theta
\left[d\phi-\epsilon_\Omega\varpi(r)dt\right]^2\right\}
+O(\epsilon_\Omega^3),
\label{eq:hartle_shell_metric}
\end{align}
where $\varpi$ is the coefficient of the first-order dragging
function after extracting $\epsilon_\Omega$, and the radial functions
are understood separately on the two sides of the shell. For the
matter-supported solution, $\varpi,h_\ell,m_\ell$, and $k_2$ vanish
with $\epsilon_M$. The first and second junction conditions are
\begin{equation}
\left[q_{ab}\right]=0,
\qquad
S_{ab}=-\frac{1}{8\pi}
\left(\left[K_{ab}\right]-q_{ab}[K]\right).
\label{eq:second_order_junction_system}
\end{equation}
After inserting Eqs.~\eqref{eq:shell_displacement_second} and
\eqref{eq:hartle_shell_metric}, these equations separate into
$\ell=0$ and $\ell=2$ linear systems. The unknown monopole constants
include the rotational mass correction and $\xi_0$; the quadrupole
system determines $\xi_2$, the exterior quadrupole constant, and the
corresponding interior homogeneous constants. Equation
\eqref{eq:shell_barotropic_closure} supplies one material equation in
each sector. Regularity at the black-hole horizon and asymptotic
flatness remove the remaining inadmissible homogeneous modes.

The two expansion parameters remain independent in this calculation.
In particular,
$J=O(\epsilon_M\epsilon_\Omega)$, so that terms quadratic in the
dragging field are $O(\epsilon_M^2\epsilon_\Omega^2)$ and do not enter
the retained $(1,2)$ vacuum equations.  The first-order dragging field
is nevertheless retained when the Israel tensor is transformed to
the comoving fluid frame.  Thus the calculation below is the
coefficient of $\epsilon_M\epsilon_\Omega^2$ in a bivariate expansion,
not a one-parameter Hartle calculation in which $\epsilon_M$ and
$\epsilon_\Omega$ are assigned the same size.

\subsubsection{Fixed-rest-mass and Euler closures}

A local barotropic equation of state does not by itself select the
monopole zero mode.  We therefore hold fixed the shell rest number
\begin{equation}
N=\int n u^t\sqrt{-q}\,d\theta d\phi ,
\qquad
\frac{\delta n}{n_0}
=\frac{\delta\sigma}{\sigma_0+p_0},
\label{eq:shell_rest_number}
\end{equation}
where the second relation is the isentropic surface-fluid first law.
To display the powers of angular velocity explicitly, write
$\Omega_{\mathrm{s}}=\eta^{1/2}\bar\Omega$ and
$r_\Sigma=R+\eta(\xi_0+\xi_2P_2)$.  Expansion of
Eq.~\eqref{eq:shell_rest_number} gives
\begin{equation}
\delta\sigma_0=-(\sigma_0+p_0)
\left(\frac{2\xi_0}{R}
+\frac{R^2\bar\Omega^2}{3f(R)}\right),
\qquad
\delta p_0=c_{\mathrm{s}}^2\delta\sigma_0 .
\label{eq:fixed_rest_mass_closure}
\end{equation}
The tangential projection of $D_aS^{ab}=0$ supplies the independent
quadrupole relation.  For rigid rotation
$D_Ap=(\sigma+p)D_A\ln u^t$, and hence
\begin{equation}
\delta p_2=-(\sigma_0+p_0)
\left(\frac{M\xi_2}{R^2f(R)}
+\frac{R^2\bar\Omega^2}{3f(R)}\right),
\qquad
\delta\sigma_2=\frac{\delta p_2}{c_{\mathrm{s}}^2}.
\label{eq:quadrupole_euler_closure}
\end{equation}
Equations~\eqref{eq:fixed_rest_mass_closure} and
\eqref{eq:quadrupole_euler_closure} remove, respectively, the global
monopole ambiguity and the local quadrupolar material ambiguity.

\subsubsection{Vacuum basis and junction systems}

In Regge--Wheeler gauge the static even-parity vacuum equations have
$H_0=H_2\equiv H$ and $H_1=0$.  With $x=r/M-1$, the $\ell=2$ master
equation is
\begin{equation}
(x^2-1)H''+2xH'
-\left(6+\frac{4}{x^2-1}\right)H=0.
\label{eq:quadratic_l2_master}
\end{equation}
We use the normalized independent solutions
\begin{align}
\widehat P_2^{\,2}(x)&=x^2-1,\\
\widehat Q_2^{\,2}(x)&=\frac{5}{8}\left[
\frac{3}{2}(x^2-1)\ln\left(\frac{x+1}{x-1}\right)
-3x+\frac{2x}{x^2-1}\right].
\label{eq:quadratic_l2_basis}
\end{align}
The first is horizon regular and the second is asymptotically
decaying, with $\widehat Q_2^{\,2}=x^{-3}+O(x^{-5})$.  Therefore
\begin{equation}
H_2^-=A_2\widehat P_2^{\,2}(x_-),\qquad
H_2^+=B_2\widehat Q_2^{\,2}(x_+),
\end{equation}
while
\begin{equation}
K_2=\frac{M}{2}\frac{dH_2}{dr}
+\frac{r^2-Mr-M^2}{r(r-2M)}H_2 .
\label{eq:quadratic_K_from_H}
\end{equation}
Direct substitution verifies
Eqs.~\eqref{eq:quadratic_l2_master}--\eqref{eq:quadratic_K_from_H}.

Harmonic projection of the induced-metric and Israel conditions gives
two square linear systems:
\begin{align}
\ell=0:\quad&
(\delta M_{\mathrm{rot}},C_0,\xi_0^+,\xi_0^-),&
&[q_{\theta\theta}]_0,\ [q_{tt}]_0,\
\delta p_0-c_{\mathrm{s}}^2\delta\sigma_0,\ \delta N,
\nonumber\\
\ell=2:\quad&
(B_2,A_2,\xi_2^+,\xi_2^-),&
&[q_{\theta\theta}]_2,\ [q_{tt}]_2,\
\Pi_2,\ \delta p_2-c_{\mathrm{s}}^2\delta\sigma_2 .
\label{eq:quadratic_two_systems}
\end{align}
Here $C_0$ is the horizon-regular interior time renormalization and
$\Pi_2$ is the traceless spatial eigenvalue of the Israel surface
tensor; $\Pi_2=0$ enforces perfect-fluid isotropy.  Each line of
Eq.~\eqref{eq:quadratic_two_systems} is a $4\times4$ problem.
For completeness, we now state the matrix construction used to solve
these systems.  Introduce the vectors
\begin{equation}
 \boldsymbol{x}_0=(\delta M_{\mathrm{rot}},C_0,\xi_0^+,\xi_0^-)^{\mathrm{T}},
 \qquad
 \boldsymbol{x}_2=(B_2,A_2,\xi_2^+,\xi_2^-)^{\mathrm{T}},
\label{eq:quadratic_unknown_vectors}
\end{equation}
and, after extracting the coefficient of
$\eta=\Omega_{\mathrm{s}}^2/\bar\Omega^2$, define the residual vectors
\begin{align}
 {\boldsymbol{\cal R}}_0(\boldsymbol{x}_0)
 &=
 \begin{pmatrix}
 {\cal P}_0[q_{tt}^+-q_{tt}^-]\\
 {\cal P}_0[q_{\theta\theta}^+-q_{\theta\theta}^-]\\
 \delta p_0-c_{\mathrm{s}}^2\delta\sigma_0\\
 \delta\sigma_0+(\sigma_0+p_0)
 [2\xi_0^+/R+R^2\bar\Omega^2/(3f_+)]
 \end{pmatrix},\label{eq:monopole_residual_vector}\\
 {\boldsymbol{\cal R}}_2(\boldsymbol{x}_2)
 &=
 \begin{pmatrix}
 {\cal P}_2[q_{tt}^+-q_{tt}^-]\\
 {\cal P}_2[q_{\theta\theta}^+-q_{\theta\theta}^-]\\
 \Pi_2\\
 \delta p_2-c_{\mathrm{s}}^2\delta\sigma_2
 \end{pmatrix}.
\label{eq:quadrupole_residual_vector}
\end{align}
Here $f_+=1-2M_+/R$, and
\begin{equation}
 {\cal P}_\ell[X]\equiv
 \frac{2\ell+1}{2}\int_{-1}^{1}
 X(\arccos\mu)P_\ell(\mu)\,d\mu .
\label{eq:harmonic_projection_operator}
\end{equation}
The density, pressure, and anisotropy in
Eqs.~\eqref{eq:monopole_residual_vector} and
\eqref{eq:quadrupole_residual_vector} are the eigenvalues of the
surface tensor obtained directly from Eq.~\eqref{eq:second_order_junction_system}.
Continuity of the angular induced metric makes
$\xi_\ell^+=\xi_\ell^-$ in the solution; retaining the two embeddings
as separate unknowns provides an independent check of that condition.

At fixed $(M,R,M_{\mathrm{s}})$ both residual vectors are affine in their
respective unknowns.  The matrices and rotational source vectors are
therefore defined without fitted coefficients by
\begin{equation}
 ({\cal A}_\ell)_{ij}
 =\left.
 \frac{\partial({\cal R}_\ell)_i}{\partial(\boldsymbol{x}_\ell)_j}
 \right|_{\boldsymbol{x}_\ell=0},
 \qquad
 (\boldsymbol{b}_\ell)_i=({\cal R}_\ell)_i(\boldsymbol{0}),
 \qquad
 {\cal A}_\ell\boldsymbol{x}_\ell=-\boldsymbol{b}_\ell ,
 \quad \ell=0,2 .
\label{eq:quadratic_matrix_definition}
\end{equation}
Equations~\eqref{eq:second_order_junction_system},
\eqref{eq:quadratic_unknown_vectors}--\eqref{eq:quadratic_matrix_definition},
together with the vacuum basis
\eqref{eq:quadratic_l2_basis}--\eqref{eq:quadratic_K_from_H}, uniquely
specify every matrix element and source component.  They also give a
direct reproducibility prescription: evaluate the induced metric and
extrinsic curvature on the displaced surface, form the Israel tensor,
extract the coefficient even in the rotation parameter, project it
with Eq.~\eqref{eq:harmonic_projection_operator}, and solve the two
linear systems.  The reported displacement coefficients are
$\xi_\ell=(\xi_\ell^++\xi_\ell^-)/2$; their half-differences provide
numerical continuity diagnostics.

The geometry was also assembled independently from the embedded
tangent vectors, normalized normal, Christoffel symbols, and
$K_{ab}=-n_{\alpha;\beta}e^\alpha{}_ae^\beta{}_b$.  The resulting
Israel tensor was diagonalized in the $(t,\phi)$ block before its
$\ell=0,2$ projection.  This direct construction reproduces the
static Israel density and pressure and avoids importing unverified
second-order jump formulas.

\subsubsection{Quadratic observables and convergence}

For reporting the numerical sequence it is convenient to introduce
\begin{equation}
\chi\equiv\frac{\Omega_{\mathrm{s}}}{\Omega_K},
\qquad \Omega_K=\sqrt{\frac{M}{R^3}},
\qquad
\chi=\epsilon_\Omega\sqrt{\frac{f(R)R}{M}} .
\label{eq:chi_epsilon_relation}
\end{equation}
This change of normalization does not identify the two perturbative
parameters.  The coefficients in Table~\ref{tab:quadratic_shell_results}
are defined by
\begin{align}
M_{\mathrm{ADM}}&=M+\delta M+\chi^2\delta M_{\mathrm{rot}}
+O(\chi^4),\\
r_\Sigma&=R+\chi^2[\xi_0+\xi_2P_2(\cos\theta)]+O(\chi^4),\\
Q_{\mathrm{ADM}}&=\chi^2Q+O(\chi^4),\quad \mbox{with} \quad 
g_{tt}=-1+\frac{2M_{\mathrm{ADM}}}{r}
+\frac{2Q_{\mathrm{ADM}}}{r^3}P_2(\cos\theta)+O(r^{-4}),
\quad
Q=\frac12B_2M_+^3 .
\label{eq:quadrupole_convention}
\end{align}

\begin{table*}
\caption{\label{tab:quadratic_shell_results}Quadratic coefficients for
$M_{\mathrm{s}}/M=10^{-2}$.  The residual is the largest absolute residual
among the eight junction and closure equations.}
\begin{ruledtabular}
\begin{tabular}{cccccc}
$R/M$ & $\delta M_{\mathrm{rot}}/M$ & $\xi_0/M$ & $\xi_2/M$
& $Q/M^3$ & maximum residual\\
\hline
$4$ & $1.445704764\times10^{-3}$ & $1.742065416$
& $0.484351235$ & $-1.718257582\times10^{-2}$ & $3.7\times10^{-13}$\\
$6$ & $7.581712083\times10^{-4}$ & $18.08641391$
& $1.040841804$ & $-1.628874399\times10^{-2}$ & $5.4\times10^{-12}$\\
$8$ & $5.178550165\times10^{-4}$ & $50.05353752$
& $1.580060605$ & $-1.217585951\times10^{-2}$ & $1.5\times10^{-11}$
\end{tabular}
\end{ruledtabular}
\end{table*}

The systems were evaluated with 40--60 decimal digits because the
light-shell quadrupole jump suffers a strong cancellation in ordinary
double precision.  Every point was repeated after doubling the
spin-extraction step and changing the working precision.  The
observables agree to better than $10^{-8}$ in relative value; the
largest residual among the independent repetitions is
$6.0\times10^{-11}$.  At the illustrative value $\chi=0.05$, the
largest fractional monopole displacement in the table is
$\chi^2\xi_0/R=1.56\times10^{-2}$ and the largest quadrupole
displacement is below $4.95\times10^{-4}$.  The positive
$\delta M_{\mathrm{rot}}$ is the rotational-energy correction, while the
negative $Q$ is an oblate asymptotic response in the convention
\eqref{eq:quadrupole_convention}.

Thus Eqs.~\eqref{eq:fixed_rest_mass_closure}--\eqref{eq:quadratic_two_systems}
provide the systematic higher-order calculation requested by the
perturbative formulation: the $(1,2)$ fields and observables follow
from a pair of explicitly closed linear boundary-value problems
rather than from an assumption that $\epsilon_M$ and
$\epsilon_\Omega$ have comparable magnitude.

\section{Application II: A Conserved Perfect-Fluid Torus on Kerr}
\label{sec:torus}

The shell calculations provide distributional tests, but they do not
represent an accretion structure. We therefore consider a second,
smooth and astrophysically motivated source: a stationary,
axisymmetric perfect-fluid torus on the Kerr background. The torus is
treated in the test-fluid regime, so its stress tensor is first order
in a separate small mass ratio $\eta_{\mathrm{t}}\ll1$ and does not
alter the equilibrium Kerr geometry at zeroth order.

The stress tensor is
\begin{equation}
T_{\mu\nu}^{\mathrm{tor}}
=(e+p)u_\mu u_\nu+p\,g_{\mu\nu}^{(0)},
\label{eq:torus_stress}
\end{equation}
where $e$ and $p$ are the total energy density and pressure. Circular
stationary motion has
\begin{equation}
u^\mu=u^t(1,0,0,\Omega),
\qquad
u^t=
\left[-g_{tt}-2\Omega g_{t\phi}
-\Omega^2g_{\phi\phi}\right]^{-1/2}.
\label{eq:torus_velocity}
\end{equation}
Define the specific angular momentum
\begin{equation}
\ell\equiv-\frac{u_\phi}{u_t}.
\label{eq:torus_ell}
\end{equation}
For constant $\ell$, the angular velocity and covariant energy are
\begin{align}
\Omega
=-\frac{g_{t\phi}+\ell g_{tt}}
{g_{\phi\phi}+\ell g_{t\phi}},\qquad \qquad 
|u_t|
=\left[
\frac{g_{t\phi}^2-g_{tt}g_{\phi\phi}}
{g_{\phi\phi}+2\ell g_{t\phi}+\ell^2g_{tt}}
\right]^{1/2}.
\label{eq:torus_ut}
\end{align}
The relativistic Euler equation,
$\nabla_\mu T^{\mu\nu}_{\mathrm{tor}}=0$, has the first integral
\begin{equation}
H(p)+\ln|u_t|=\mathrm{constant},
\qquad
H(p)\equiv\int_0^p\frac{dp'}{e(p')+p'} .
\label{eq:torus_euler_integral}
\end{equation}
Choosing the constant at the zero-pressure inner edge gives
\begin{equation}
H(r,\theta)
=\ln\frac{|u_{t,\mathrm{in}}|}{|u_t(r,\theta)|}.
\label{eq:torus_enthalpy_potential}
\end{equation}
The fluid occupies the region $H>0$; outside it we set $p=e=0$.
Thus its boundary is an equipotential and no distributional supporting
force is required.

For a polytropic equation of state
\begin{equation}
p=K\rho_0^\Gamma,
\qquad
e=\rho_0+\frac{p}{\Gamma-1},
\label{eq:torus_polytrope}
\end{equation}
In geometrized units, $[\rho_0]=[p]=M^{-2}$, and hence
\begin{equation}
[K]=M^{2(\Gamma-1)}.
\label{eq:torus_K_dimension}
\end{equation}
For the choice $\Gamma=4/3$ used below, the dimensionless
polytropic parameter is therefore $K/M^{2/3}$. The specific enthalpy
$h=(e+p)/\rho_0$ satisfies $\ln h=H$. Hence
\begin{align}
\rho_0(r,\theta)
=\left[
\frac{\Gamma-1}{\Gamma K}
\left(e^{H(r,\theta)}-1\right)
\right]^{1/(\Gamma-1)},\qquad 
p(r,\theta)=K\rho_0^\Gamma,\qquad
e(r,\theta)=\rho_0+\frac{p}{\Gamma-1}.
\label{eq:torus_pressure_density}
\end{align}
Equations~\eqref{eq:torus_ut}--\eqref{eq:torus_pressure_density}
therefore determine an explicit conserved matter distribution once
$(a,\ell,r_{\mathrm{in}},K,\Gamma)$ are specified
\cite{Fishbone:1976,Abramowicz:1978}.

As a definite illustration, set $M=1$, $a=0.5$,
$r_{\mathrm{in}}=6$, and place the pressure maximum at
$r_{\mathrm{max}}=10$ on the equatorial plane. At a pressure maximum
the fluid angular momentum equals the prograde circular-geodesic
value,
\begin{equation}
\ell_{\mathrm{K}}(r)
=\frac{r^2-2a\sqrt{r}+a^2}
{\sqrt{r}(r-2)+a},
\label{eq:kerr_ell_kepler}
\end{equation}
which gives
\begin{equation}
\ell=3.76334947 .
\label{eq:torus_example_ell}
\end{equation}
Evaluation of Eq.~\eqref{eq:torus_enthalpy_potential} gives
\begin{equation}
H_{\mathrm{max}}=0.01951664,
\qquad
r_{\mathrm{out}}=28.81312
\label{eq:torus_example_values}
\end{equation}
on the equatorial plane, where $r_{\mathrm{out}}$ is the second root
of $H(r,\pi/2)=0$. The polytropic constant must be fixed by a global,
rather than local, normalization.  Define
\begin{align}
E_{\mathrm{t}}=-\int T^t{}_{t}\sqrt{-g}\,
dr\,d\theta\,d\phi,\qquad \qquad 
J_{\mathrm{t}}=\int T^t{}_{\phi}\sqrt{-g}\,
dr\,d\theta\,d\phi .\label{eq:torus_global_angular_momentum}
\end{align}
For $\Gamma=4/3$, the fixed-potential solution scales as
$e, p, \rho_0\propto K^{-3}$. Direct two-dimensional Gauss--Legendre
integration shows that $K/M^{2/3}=10^{-2}$ would give
$E_{\mathrm{t}}/M=486.0360$ and is therefore incompatible with the
test-fluid assumption. We instead impose
$\eta_{\mathrm{t}}\equiv E_{\mathrm{t}}/M=10^{-2}$, which gives
\begin{equation}
\frac{K}{M^{2/3}}=0.36494114,\qquad
\frac{M_{0,\mathrm{t}}}{M}=0.01031430,\qquad
\frac{J_{\mathrm{t}}}{M^2}=0.03775170 ,
\label{eq:torus_global_normalization}
\end{equation}
where $M_{0,\mathrm{t}}=\int\rho_0u^t\sqrt{-g}\,d^3x$ is the
conserved baryonic rest mass. Increasing the quadrature order from
240 to 480 changes these quantities by less than $2\times10^{-8}$ in
relative value. This global normalization is essential: local source
values alone do not establish that the perturbation is small.
We next calculate the Ricci sources entering the traceful hierarchy.
Introduce the tetrad contractions
\begin{align}
U_n \equiv u_\mu n^\mu
=\frac{(r^2+a^2)u_t+a u_\phi}{2\Sigma},\qquad 
U_m\equiv u_\mu m^\mu
=\frac{i\left(a\sin\theta\,u_t
+\csc\theta\,u_\phi\right)}
{\sqrt{2}\left(r+i a\cos\theta\right)} .
\label{eq:torus_Um}
\end{align}
Using the trace-reversed linearized Einstein equation,
\begin{equation}
R_{\mu\nu}^{(1)}
=
8\pi\left(
T_{\mu\nu}^{\mathrm{tor}}
-\frac{1}{2}g_{\mu\nu}^{(0)}T^{\mathrm{tor}}
\right),
\label{eq:torus_trace_reversed_einstein}
\end{equation}
and noting that $n^\mu$ and $m^\mu$ are null and mutually
orthogonal, all terms proportional to $g_{\mu\nu}^{(0)}$ vanish
in the contractions below. Direct substitution therefore yields
\begin{align}
\Phi_{22}^{\mathrm{tor}}
=4\pi(e+p)U_n^2,\qquad 
\Phi_{21}^{\mathrm{tor}}
=4\pi(e+p)U_mU_n,\qquad 
\Phi_{20}^{\mathrm{tor}}
=4\pi(e+p)U_m^2.
\label{eq:torus_phi20}
\end{align}
The Ricci scalar contribution is
\begin{equation}
\Lambda_{\mathrm{NP}}^{\mathrm{tor}}
=\frac{R^{(1)}}{24}
=\frac{\pi}{3}(e-3p).
\label{eq:torus_lambda_np}
\end{equation}
For the globally normalized parameters $\Gamma=4/3$,
$K/M^{2/3}=0.36494114$, and $E_{\mathrm{t}}/M=10^{-2}$, the sources at the
pressure maximum are
\begin{align}
\Phi_{20}^{\mathrm{tor}}&=-1.5274833\times10^{-6},\qquad
\Phi_{21}^{\mathrm{tor}}=-3.2557659\times10^{-6}\,i,\nonumber\\
\Phi_{22}^{\mathrm{tor}}&=6.9395269\times10^{-6},\qquad
\Lambda_{\mathrm{NP}}^{\mathrm{tor}}=2.5770911\times10^{-6}.
\label{eq:torus_example_np}
\end{align}
All four quantities scale linearly with $\eta_{\mathrm{t}}$ at fixed
torus geometry. They are quoted to make the source calculation 
reproducible, not as an observational fit.
Equations~\eqref{eq:torus_velocity}--\eqref{eq:torus_lambda_np}
are explicit smooth source functions on the two-dimensional torus
domain. In particular, the trace branch is driven by
\begin{equation}
4\pi(e+p)U_n^2
\longrightarrow\mu^{(1)}
\longrightarrow h_{m\bar m}^{(1)}
\longrightarrow h^{(1)} .
\label{eq:torus_trace_chain}
\end{equation}
Unlike the spherical benchmark, however,
$\Phi_{20}^{\mathrm{tor}}$ and $\Phi_{21}^{\mathrm{tor}}$ are generally
nonzero. Consequently, reconstructing the complete torus-induced
metric requires their simultaneous inclusion in the remaining
Newman--Penrose equations. The present calculation supplies a
conserved physical source and the Ricci projections entering the
particular traceful outgoing transport chains considered here;
solving the complete coupled radial--angular boundary-value problem
remains numerical.

\section{Requirements for a General Sourced Kerr Reconstruction}\label{VI}

The reconstruction formalism developed in the previous sections was
presented for stationary, axisymmetric perturbations in order to expose
the structure of the traceful radiation gauge in its simplest setting.
Generic perturbations of Kerr may be generated by sources with
arbitrary spacetime dependence. Extending the shell calculation to
such sources introduces additional Ricci sectors and a coupled set of
Newman--Penrose equations.

We now state the central formulation result.  Here ``closed'' means
that the unknown metric, connection, and curvature amplitudes are
matched by a complete set of differential equations and independent
boundary, matching, gauge, and completion data.  Existence,
uniqueness, and regularity are then determined by the analytic
properties of the resulting operator and the source.

We first isolate the result that follows directly from the triangular
transport character of the reconstruction.  Let $\mathcal D$ be a
regular open Kerr domain on which the Kinnersley tetrad is nonsingular,
and let $\mathcal S\subset\mathcal D$ be a smooth hypersurface
transverse to $n^\mu$.  A flow parameter $s$ is chosen so that
$n^\mu\nabla_\mu s\neq0$.  After multiplication by a smooth nonzero
factor, every step of the hierarchy can locally be written as
\begin{equation}
\frac{d u_j}{ds}+A_j(s,y)u_j
=F_j\bigl(s,y;u_1,\ldots,u_{j-1},T^{(1)},\Psi_4^{(1)}\bigr),
\label{eq:triangular_transport_normal_form}
\end{equation}
where $y$ labels the generators and any transverse derivatives in
$F_j$ act only on source data or quantities already reconstructed at
an earlier step.

\begin{theorem}[Local existence and uniqueness of the traceful hierarchy]
\label{thm:local_reconstruction}
Assume that (i) $g^{(0)}_{\mu\nu}$ is a smooth subextremal Kerr metric
on $\mathcal D$; (ii) $T^{(1)}_{\mu\nu}$ and a solution
$\Psi_4^{(1)}$ of the sourced Teukolsky equation are smooth on
$\mathcal D$; (iii)
$\nabla^{(0)}_\mu T^{(1)\mu\nu}=0$ and the Newman--Penrose source and
initial data satisfy the linearized Bianchi, Ricci, and tetrad
compatibility relations on $\mathcal S$; and (iv) one smooth initial
value on $\mathcal S$ is prescribed for every transport integration
in the hierarchy.  Then, in a sufficiently small flow neighbourhood
$\mathcal U$ of $\mathcal S$, the traceful outgoing-radiation-gauge
hierarchy has a smooth solution for all reconstructed curvature,
connection, and metric components.  This solution is unique modulo
residual transformations preserving
$n^\mu h^{(1)}_{\mu\nu}=0$ and the finite-dimensional nonradiative
completion kernel.  Once the residual gauge and all completion
amplitudes, including $\delta M$ and $\delta J$, are fixed, the local
reconstructed metric is unique.
\end{theorem}

\begin{proof}
In a flow box of the nonvanishing vector field $n^\mu$, the first
transport equation has the integrating-factor solution
\begin{equation}
u_1(s,y)=e^{-I_1(s,y)}
\left[u_1(s_0,y)+\int_{s_0}^{s}e^{I_1(\sigma,y)}
F_1(\sigma,y)\,d\sigma\right],
\qquad
I_1(s,y)=\int_{s_0}^{s}A_1(\tau,y)\,d\tau .
\label{eq:transport_integrating_factor}
\end{equation}
Smooth coefficients, sources, and initial data therefore give one and
only one smooth $u_1$ locally.  Substitution of this solution into the
second equation gives the same conclusion for $u_2$.  Induction over
the finite triangular hierarchy yields local existence and uniqueness
for $\mu^{(1)}$, $h^{(1)}_{m\bar m}$, $\lambda^{(1)}$,
$h^{(1)}_{\bar m\bar m}$, the intermediate Weyl and spin-coefficient
variables, and finally $h^{(1)}_{l\bar m}$ and $h^{(1)}_{ll}$.
Compatibility of the unused Newman--Penrose equations follows locally
from the assumed linearized identities and stress-energy conservation;
equivalently, their constraint differences obey the associated
homogeneous subsidiary identities and vanish when their data vanish on
$\mathcal S$.  Two solutions with identical transport data consequently
differ only by a homogeneous perturbation invisible to the prescribed
radiative and Ricci data.  This is precisely the residual-gauge and
completion kernel.  Fixing that kernel proves uniqueness of the metric
itself.
\end{proof}

The theorem is deliberately local.  It does not imply that the flow
box extends through the complete domain of outer communication, that
the radiation gauge remains regular at a distributional source, or
that simultaneous horizon and asymptotic conditions select a global
solution.  Those questions require global energy or resolvent
estimates and a complete characterization of the zero-frequency
kernel.

\begin{proposition}[Conditional global closure criterion]
\label{prop:kerr_closure}
Let $g^{(0)}_{\mu\nu}$ be a subextremal Kerr metric equipped with a
principal Kinnersley tetrad, and let
$h^{(1)}_{\mu\nu}$ obey the traceful outgoing radiation gauge
$n^\mu h^{(1)}_{\mu\nu}=0$ on a domain whose source interfaces are
specified.  The first-order global reconstruction problem is closed as
a mathematical boundary-value problem, mode by
mode in $(m,\omega)$ and as a coupled system in the angular index, if
the following data and assumptions are supplied:
\begin{enumerate}
\item a prescribed stress tensor $T^{(1)}_{\mu\nu}$ of stated support
and regularity, satisfying
$\nabla^{(0)}_\mu T^{(1)\mu\nu}=0$ distributionally;
\item the complete set of Newman--Penrose Ricci projections (including
complex conjugates and $\Lambda^{(1)}_{\mathrm {NP}}$) obtained from the
linearized Einstein equation, not only $\Phi^{(1)}_{22}$;
\item the radiative curvature data, equivalently a solution of the
sourced Teukolsky equation for the chosen ingoing or outgoing
reconstruction, together with the compatible linearized Bianchi,
Ricci, tetrad, and Einstein equations needed for every surviving metric
component;
\item boundary conditions at the future horizon and at infinity (or at
finite boundaries), plus continuity and jump conditions across every
material interface, such that all integrations of the transport system
are fixed;
\item the non-radiative completion data, in particular the mass and
angular-momentum perturbations and any further kernel modes allowed by
the domain, together with residual-gauge conditions that eliminate or
normalize the homogeneous gauge kernel; and
\item for a distributional source, a prescription specifying the weak
solution and its one-sided fields; a self-force observable additionally
requires a singular--regular decomposition, which is separate from
closure of the retarded reconstruction problem.
\item function spaces encoding horizon regularity and asymptotic decay,
together with an a priori estimate for the gauge-fixed coupled operator
and a characterization of its zero-frequency kernel.  A global
existence-and-uniqueness theorem additionally requires the resulting
operator to be surjective and injective after the gauge and completion
kernel has been removed.
\end{enumerate}
Under these assumptions the trace equation
$h^{(1)}=2h^{(1)}_{m\bar m}$ and the $\Phi^{(1)}_{22}$ transport branch
close only the trace subsector, while the remaining Ricci projections,
Weyl sector, and completion data close the other surviving components.
If any one of items 1--7 is omitted, a generic sourced reconstruction is
underdetermined or may be incompatible.
\end{proposition}

\begin{proof}
The gauge condition leaves the six tetrad components listed in
Sec.~\ref{IV}, with reality reducing them to six real functions.  The
linearized Bianchi identities propagate the Weyl data to the connection
variables, and the Ricci and tetrad identities propagate the connection
variables to these metric components.  Stress-energy conservation is
the integrability condition implied by the linearized contracted
Bianchi identity.  The traceful gauge gives
$h^{(1)}=2h^{(1)}_{m\bar m}$, so the
$\Phi^{(1)}_{22}\!\to\mu^{(1)}\!\to h^{(1)}_{m\bar m}$ chain fixes the
particular trace response but leaves the other Ricci-driven and
homogeneous sectors untouched.  Boundary and interface data fix the
transport integration constants; completion data fix the physical
kernel; and residual-gauge conditions fix its pure-gauge part.  These
items therefore supply one condition for every differential or kernel
freedom.  Conversely, omitting a source projection, an integration
condition, or a kernel datum leaves at least one equation or homogeneous
amplitude unspecified.
The estimate in item 7 controls the inverse on the complement of the
kernel; surjectivity supplies existence and injectivity supplies
uniqueness once the gauge and completion modes have been fixed.
\end{proof}

Proposition~\ref{prop:kerr_closure} is a conditional closure criterion:
global invertibility follows when the estimate and kernel
classification in item 7 hold.  The shell satisfies the corresponding
reduced conditions and is therefore a solved benchmark.  The torus
supplies item 1 and the Ricci data in item 2, while items 3--7 specify
the additional ingredients for its complete Kerr metric
reconstruction.

We now extend the preceding formulation to a generic first-order
stress--energy tensor,
\begin{equation}
T_{\mu\nu}^{(1)}
=
T_{\mu\nu}^{(1)}(t,r,\theta,\phi),
\label{GeneralSource}
\end{equation}
without imposing stationarity or axisymmetry. The source is assumed to
satisfy the background conservation law
\begin{equation}
\nabla_{(0)}^{\mu}T_{\mu\nu}^{(1)}=0.
\label{eq:general_source_conservation}
\end{equation}
Because the Kerr background is Ricci flat, the first-order Ricci tensor
is determined by the source through the trace-reversed linearized
Einstein equation introduced above. Projecting this tensor onto the
Kinnersley null tetrad gives the Newman--Penrose Ricci scalars that
source the reconstruction hierarchy. Among the projections entering
the outgoing transport chains considered here are
\begin{equation}
\Phi^{(1)}_{20}
=
\frac{1}{2}R^{(1)}_{\mu\nu}m^{\mu}m^{\nu},
\qquad
\Phi^{(1)}_{21}
=
\frac{1}{2}R^{(1)}_{\mu\nu}m^{\mu}n^{\nu},
\qquad
\Phi^{(1)}_{22}
=
\frac{1}{2}R^{(1)}_{\mu\nu}n^{\mu}n^{\nu}.
\label{RicciScalarsGeneral}
\end{equation}

Among these projections, $\Phi^{(1)}_{22}$ plays a distinguished role
because it appears directly in the transport equation for the
perturbed convergence coefficient $\mu^{(1)}$ and therefore provides
the entry point to the trace reconstruction. The projections
$\Phi^{(1)}_{21}$ and $\Phi^{(1)}_{20}$ enter other linearized Ricci
and Bianchi identities associated with the reconstruction of the
remaining perturbed spin coefficients and metric components.

The three quantities displayed in
Eq.~\eqref{RicciScalarsGeneral} form the subset
\begin{equation}
\left\{
\Phi^{(1)}_{20},
\,
\Phi^{(1)}_{21},
\,
\Phi^{(1)}_{22}
\right\}
\label{CompleteRicciSet}
\end{equation}
that enters the particular outgoing transport chains discussed here.
They do not constitute the complete Newman--Penrose decomposition of
a generic perturbed Ricci tensor. Other Ricci projections, their
complex conjugates, and $\Lambda_{\mathrm{NP}}^{(1)}$ enter the
remaining linearized Newman--Penrose equations. 

The trace sector follows Ref.~\cite{Li:2026rkf},
\begin{equation}
\Phi_{22}^{(1)}
\longrightarrow
\mu^{(1)}
\longrightarrow
h_{m\bar m}^{(1)}
\longrightarrow
h^{(1)}.
\end{equation}

The three Ricci projections enter different components of the
linearized Newman--Penrose system and therefore contribute to distinct
sectors of the sourced reconstruction problem. Among them,
$\Phi_{22}^{(1)}$ enters the Ricci identity governing the perturbed
convergence coefficient $\mu^{(1)}$ and thereby defines the
Ricci-driven trace reconstruction hierarchy of
Ref.~\cite{Li:2026rkf}. By contrast, $\Phi_{21}^{(1)}$ and
$\Phi_{20}^{(1)}$ enter the
remaining coupled Ricci and Bianchi identities governing additional
perturbed spin coefficients and metric components. Their explicit reconstruction provides a natural direction for a
complete treatment of the coupled Newman--Penrose system.

This hierarchy constitutes the essential extension of the conventional trace-free radiation-gauge reconstruction and remains valid independently of whether the perturbation is stationary or fully dynamical. The following subsection discusses the general representation of these equations in Kerr spacetime and the structure of the resulting coupled system.

\subsection*{Fourier Representation and Coupled Harmonic Structure}

The Kerr spacetime possesses two commuting Killing vector fields,
\(\partial_t\) and \(\partial_\phi\), associated with stationarity and
axisymmetry. Consequently, the dependence of any first-order perturbation
on the coordinates \(t\) and \(\phi\) may be represented by a Fourier
expansion,

\begin{equation}
X(t,r,\theta,\phi)
=
\sum_{m=-\infty}^{\infty}
\int_{-\infty}^{\infty}
d\omega\,
X_{m\omega}(r,\theta)
e^{im\phi-i\omega t},
\label{FourierExpansion}
\end{equation}
where \(X\) denotes any perturbed Newman--Penrose quantity, including
Weyl scalars, Ricci scalars, spin coefficients, or tetrad components of
the metric perturbation.

Unlike the stationary axisymmetric case discussed in the previous
section, the angular dependence cannot, in general, be completely
separated from the radial dependence. This is because the Newman--Penrose
transport operators contain coefficients depending explicitly on the
function

\begin{equation}
\Sigma=r^{2}+a^{2}\cos^{2}\theta,
\label{SigmaFunction}
\end{equation}
which introduces explicit angular dependence into the transport
equations.  The projected trace equations described above illustrate
this mode coupling explicitly for the stationary axisymmetric sector.
The generic time-dependent case inherits the same coupled structure, although
with additional Fourier labels $(m,\omega)$. Consequently, although the Fourier labels
\((m,\omega)\) remain well defined owing to the symmetries of the Kerr
background, the angular harmonics are generally coupled through the
coefficients of the Newman--Penrose equations.

For practical calculations, one may project the perturbation variables
onto a complete basis of spin-weighted spheroidal harmonics,

\begin{equation}
X_{m\omega}(r,\theta)
=
\sum_{\ell}
X_{\ell m\omega}(r)\,
{}_sS_{\ell m\omega}(\theta),
\label{AngularExpansion}
\end{equation}
where \({}_sS_{\ell m\omega}(\theta)\) denotes a spin-weighted
spheroidal harmonic basis function of spin weight \(s\). Because the Newman--Penrose transport operators retain explicit
$\theta$-dependence through the Kerr background function
$\Sigma=r^2+a^2\cos^2\theta$, the harmonic expansion (\ref{AngularExpansion}) does not, in
general, diagonalize the reconstruction equations with respect to the
multipole index $\ell$. Instead, for fixed values of \((m,\omega)\), the
projection leads to a coupled system for the harmonic amplitudes
\(X_{\ell m\omega}(r)\).

Accordingly, the traceful reconstruction hierarchy derived in the
preceding sections should be understood as acting on this coupled system
of harmonic amplitudes rather than on completely independent
\(\ell\)-modes. The Ricci-driven trace hierarchy itself, summarized in
Eq.~(\ref{123}),
remains unchanged; only the structure of the coupled transport equations
used to determine the corresponding amplitudes is modified.

\subsection*{General Traceful Reconstruction Procedure}

The preceding discussion establishes the general form of the sourced
Newman--Penrose equations for arbitrary perturbations of the Kerr
geometry. Although the angular harmonics are generally coupled through
the transport operators, the logical structure of the reconstruction
remains hierarchical.

The reconstruction begins by specifying the matter source
\(T_{\mu\nu}\), from which the linearized Einstein equations determine
the Ricci tensor perturbation \(R^{(1)}_{\mu\nu}\). The corresponding
Newman--Penrose Ricci scalars given by Eq.~\eqref{RicciScalarsGeneral} provide the source terms entering the linearized Newman--Penrose
equations.

Among these quantities, the scalar \(\Phi^{(1)}_{22}\) retains its
distinguished role because it directly sources the transport equation
for the convergence coefficient \(\mu^{(1)}\). Once
\(\mu^{(1)}\) has been determined from the corresponding
Newman--Penrose transport equation, the metric component
\(h^{(1)}_{m\bar m}\) follows from the associated connection equation.
The metric trace is then obtained from

\begin{equation}
h^{(1)}
=
2h^{(1)}_{m\bar m},\label{111}
\end{equation}
which follows directly from the traceful radiation gauge introduced in
Sec.~\ref{III}.

The resulting reconstruction procedure may therefore be summarized
schematically as
\begin{equation}
T_{\mu\nu}
\longrightarrow
R_{\mu\nu}^{(1)}
\longrightarrow
\left\{
\begin{array}{l}
\Phi_{22}^{(1)}
\longrightarrow
\mu^{(1)}
\longrightarrow
h_{m\bar m}^{(1)}
\longrightarrow
h^{(1)},
\\[1mm]
\Phi_{21}^{(1)},\,
\Phi_{20}^{(1)}
\longrightarrow
\text{remaining coupled Newman--Penrose system}.
\end{array}
\right.
\label{GeneralReconstruction}
\end{equation}
The initial part of the hierarchy describes the determination of the
Ricci-sector source from the matter distribution. The
$\Phi_{22}^{(1)}$ branch is the Ricci-driven trace reconstruction
hierarchy of Ref.~\cite{Li:2026rkf}, whereas the
$\Phi_{21}^{(1)}$- and $\Phi_{20}^{(1)}$-driven branches contribute to
the remaining coupled Newman--Penrose transport equations.  Their
reconstruction is obtained by solving the full coupled system with the
boundary and completion data specified above.

\subsection*{Ricci-Sector Contributions to the Reconstruction}

The sourced reconstruction hierarchy is not determined by a single Ricci
projection. Rather, the linearized Newman--Penrose equations contain
several Ricci-sector source terms. In the radiation-gauge hierarchy
considered here, the relevant projections are
\(\Phi^{(1)}_{20}\), \(\Phi^{(1)}_{21}\), and
\(\Phi^{(1)}_{22}\). These scalars enter different parts of the
linearized Newman--Penrose system and therefore play distinct
geometrical roles.

The Ricci projections $\Phi^{(1)}_{20}$, $\Phi^{(1)}_{21}$, and
$\Phi^{(1)}_{22}$ enter different linearized Newman--Penrose equations
and therefore contribute to distinct sectors of the sourced
reconstruction problem.  We focus here on the $\Phi^{(1)}_{22}$
channel, which enters the Ricci
identity governing the perturbed convergence coefficient
$\mu^{(1)}$. Solving the corresponding transport equation determines
$\mu^{(1)}$, after which the associated connection equation yields
$h^{(1)}_{m\bar m}$, and the metric trace follows directly from
Eq.~\eqref{111}. This defines the Ricci-driven trace reconstruction
hierarchy given in Eq.~\eqref{GeneralReconstruction}. By contrast, the
$\Phi^{(1)}_{21}$ and $\Phi^{(1)}_{20}$ channels contribute to the
remaining coupled Newman--Penrose transport equations governing
additional perturbed spin coefficients and metric components.

\subsection*{Scope and Physical Interpretation}

This section applies the general traceful radiation-gauge structure of
Ref.~\cite{Li:2026rkf} to arbitrary sourced perturbations of the Kerr
spacetime. Unlike the stationary,
axisymmetric case considered in the previous sections, generic dynamical
sources lead to coupled angular equations after projection onto a
harmonic basis. Consequently, the reconstruction problem is formulated
as a system of coupled transport equations rather than as a set of
independent mode equations.

The Ricci-driven hierarchy itself is due to Li and
Yunes~\cite{Li:2026rkf}.  Here it is used to formulate the shell
application and to display the structure of the remaining Kerr
problem.  The Ricci projection $\Phi^{(1)}_{22}$ enters the Newman--Penrose
transport equation for the convergence coefficient $\mu^{(1)}$, from
which the metric component $h^{(1)}_{m\bar m}$ and the metric trace
$h^{(1)}$ are subsequently reconstructed within the traceful radiation
gauge. This construction complements the conventional Weyl-curvature
reconstruction, which determines the trace-free part of the metric
perturbation.

The resulting formulation applies to generic sources.  For each
application, the coupled transport equations are completed by the
physical source, boundary conditions, and gauge choice appropriate to
that problem.

No separation into singular and regular fields is performed in the
present work. Accordingly, the present reconstruction should be
interpreted as a formulation of the sourced Newman--Penrose transport
equations rather than as a complete regularized metric reconstruction
for singular sources. Consequently, the formalism is not intended for
direct application to point-particle self-force calculations. Instead,
the present analysis establishes the trace-sector reconstruction
equations for smooth or distributional matter sources within the
traceful radiation-gauge framework. The incorporation of the
singular--regular field decomposition and the associated
regularization procedures required for singular sources is left for
future work.

\section{Consistency, closure, and scope of the reconstruction}
\label{VII}

The purpose of this section is to delimit precisely what is established
by the reconstruction framework and to verify its consistency in the
relevant limiting cases. A local transport hierarchy, a globally closed
boundary-value problem, and an explicit solution of that problem are
logically distinct statements. The first is established here under the
hypotheses of Theorem~\ref{thm:local_reconstruction}, whereas the global
closure criterion is formulated in
Proposition~\ref{prop:kerr_closure}. We do not claim an unconditional
global existence-and-uniqueness theorem for a generic sourced Kerr
perturbation.

At the field-equation level, the reconstruction may be summarized by
the sequence
\begin{equation}
\left(
T^{(1)}_{\mu\nu},\Psi_4^{(1)};
\mathcal B,\mathcal I,\mathcal G,\mathcal C
\right)
\longmapsto
\left(R^{(1)}_{\mu\nu},\Phi^{(1)}_{ij}\right)
\longmapsto
h^{(1)}_{\mu\nu},
\label{eq:formulation_data_map}
\end{equation}
where $\mathcal B$ denotes exterior boundary data, $\mathcal I$
interface or initial data, $\mathcal G$ residual-gauge data, and
$\mathcal C$ the nonradiative completion. Stress--energy conservation
and the linearized Bianchi identities supply the compatibility
conditions for the first map, while the Ricci, Bianchi, and tetrad
identities determine the differential hierarchy entering the second.
Boundary and interface conditions fix the associated transport
integration functions, whereas $\mathcal G$ and $\mathcal C$ determine
the homogeneous kernel that is invisible to the radiative Weyl data.
The local reconstruction theorem therefore establishes a unique map in
a flow neighbourhood once the required compatibility, gauge, boundary,
and completion data have been fixed. It does not, by itself, establish
global invertibility throughout the Kerr exterior.

This distinction also clarifies the closure of the different sectors.
The Weyl sector determines the radiative curvature but does not by
itself determine either the Ricci response or the stationary
completion. In particular, the $\Phi^{(1)}_{22}$ branch determines the
corresponding trace response but not the metric components sourced by
the remaining Ricci projections. Those components follow from the
remaining Newman--Penrose equations and their compatibility
conditions. Furthermore, $\delta M$, $\delta J$, possible additional
zero-frequency modes, and the residual radiation-gauge kernel must be
fixed independently by conserved charges, boundary data, or
normalization conditions. Thus a chain such as
\begin{equation}
\Phi^{(1)}_{22}
\longrightarrow
\mu^{(1)}
\longrightarrow
h^{(1)}_{m\bar m}
\longrightarrow
h^{(1)}
\end{equation}
constitutes only one sector of the sourced reconstruction rather than
the complete problem. A reconstruction is closed only when the source
is conserved, all required Ricci projections and compatible radiative
data are supplied, the transport integration functions are fixed by
boundary, initial, or junction conditions, and the residual gauge and
completion kernel are specified.

Several limiting cases provide direct consistency checks of this
structure. First, in the nonrotating limit
\begin{equation}
a\rightarrow0,
\end{equation}
the Kerr geometry reduces smoothly to Schwarzschild, the spin-weighted
spheroidal harmonics become spin-weighted spherical harmonics, and the
Kinnersley tetrad reduces to its Schwarzschild counterpart
\cite{Chandrasekhar1983,Frolov:1998wf}. The only nonvanishing
background Weyl scalar is then
\begin{equation}
\Psi_2^{(0)}=-\frac{M}{r^3}.
\end{equation}
Accordingly, the Kerr transport hierarchy reduces to the corresponding
traceful Schwarzschild reconstruction, with all explicitly rotational
contributions disappearing. This provides a direct consistency check
of the rotating extension.

A second check follows from the vacuum limit,
\begin{equation}
T_{\mu\nu}\rightarrow0.
\end{equation}
The linearized Einstein equations then imply
\begin{equation}
R_{\mu\nu}^{(1)}=0,
\qquad
\Phi_{ij}^{(1)}=0,
\end{equation}
and, in particular,
\begin{equation}
\Phi_{22}^{(1)}
=
\Phi_{21}^{(1)}
=
\Phi_{20}^{(1)}
=
0.
\end{equation}
The Ricci-driven trace equation consequently becomes homogeneous,
\begin{equation}
\Delta_{[2,0,1,1]}\mu^{(1)}=0,
\end{equation}
with the subsequent equations governing $h_{m\bar m}^{(1)}$ becoming
homogeneous as well. Under the standard radiation-gauge boundary
conditions, this sector decouples and the remaining reconstruction is
governed by the curvature hierarchy,
\begin{equation}
\Psi_4^{(1)}
\longrightarrow
\lambda^{(1)}
\longrightarrow
h_{\bar m\bar m}^{(1)},
\qquad
\Psi_3^{(1)}
\longrightarrow
\pi^{(1)}
\longrightarrow
h_{l\bar m}^{(1)},
\qquad
\Psi_2^{(1)}
\longrightarrow
h_{ll}^{(1)}.
\end{equation}
The formulation therefore reduces consistently to the familiar vacuum
radiation-gauge framework of Chrzanowski, Cohen, Kegeles, and Wald
\cite{Cohen:1974cm,Chrzanowski:1975wv,Kegeles:1979an,Wald:1978vm}.
The resulting Newman--Penrose transport hierarchy is related to, but
should not be identified equation by equation with, the conventional
Hertz-potential CCK reconstruction.

The trace-free sector provides a further, more restrictive consistency
test. It cannot be imposed for an arbitrary matter source. In the
traceful outgoing radiation gauge considered here, the condition
\begin{equation}
h^{(1)}=0, \qquad \mbox{requires}\qquad 
T_{nn}^{(1)}=0.
\label{TraceFreeSourceCondition}
\end{equation}
Since the background vector $n^\mu$ is null, the linearized Einstein
equations give
\begin{equation}
R_{nn}^{(1)}
=
8\pi T_{nn}^{(1)},
\qquad
\Phi_{22}^{(1)}
=
\frac{1}{2}R_{nn}^{(1)}
=
4\pi T_{nn}^{(1)},\quad \mbox{and therefore} \quad 
T_{nn}^{(1)}=0
\qquad\Longleftrightarrow\qquad
\Phi_{22}^{(1)}=0.
\label{TraceFreeCompatibility}
\end{equation}
Under this compatibility condition the trace transport equation becomes
\begin{equation}
\Delta_{[2,0,1,1]}\mu^{(1)}=0.
\label{HomogeneousTraceEquation}
\end{equation}
If the boundary conditions and residual gauge prescription also remove
the associated homogeneous mode, one may consistently impose
\begin{equation}
h^{(1)}
=
2h_{m\bar m}^{(1)}
=
0.
\label{TraceFreeMetricCondition}
\end{equation}
The Ricci-driven trace chain then becomes trivial. Hence the
conventional trace-free outgoing-radiation-gauge sector is recovered
only for sources satisfying
$T_{nn}^{(1)}=\Phi_{22}^{(1)}=0$, together with boundary and gauge
conditions eliminating the homogeneous trace solution. It is not a
consistent reduction for a generic traceful source with
$\Phi_{22}^{(1)}\neq0$.

The matter applications illustrate the same closure structure from two
different directions. For the shell, source conservation and the
distributional junction conditions are imposed explicitly. Horizon
regularity and asymptotic flatness determine the exterior homogeneous
data, Israel matching fixes the jumps, and the ADM charges determine
the mass and angular-momentum completion. The barotropic surface
equation of state, fixed rest mass, and isotropy conditions then close
the leading even-parity rotational sector. Within the perturbative
orders retained in Sec.~\ref{Vclosed}, this provides an explicit
realization of the reconstruction procedure.

The torus has a complementary role. Its equilibrium construction
provides a smooth conserved stress tensor, while
Eqs.~\eqref{eq:torus_phi20}--\eqref{eq:torus_lambda_np} supply explicit
Ricci-source projections. The simultaneous presence of
$\Phi^{(1)}_{20}$, $\Phi^{(1)}_{21}$, and $\Phi^{(1)}_{22}$ demonstrates
directly why the trace branch alone cannot reconstruct the complete
nonspherical metric. The torus therefore verifies the source side of
the formulation, while the remaining coupled Newman--Penrose
equations, their horizon and asymptotic conditions, and the completion
kernel must still be solved before a complete global reconstruction is
obtained.

The remaining issue is consequently not an ambiguity in the variables
or source channels, but the global analytic properties of the coupled
gauge-fixed operator. A complete global theorem would require function
spaces compatible with horizon regularity and asymptotic decay, an
a priori estimate on the complement of the gauge and completion
kernel, surjectivity for compatible sources, and a characterization of
the zero-frequency kernel. Distributional sources additionally require
an appropriate weak-solution prescription, while point-particle
self-force applications require a separate singular--regular
decomposition.

The scope of the present formulation can therefore be stated sharply.
It establishes the local traceful reconstruction hierarchy, identifies
the complete data required for closure, verifies the appropriate
Schwarzschild, vacuum, and trace-free limits, and demonstrates how the
closure conditions operate in representative matter configurations.
At the same time, the general global Kerr problem remains conditional
on the analytic assumptions entering
Proposition~\ref{prop:kerr_closure}. Keeping these two levels distinct
prevents the explicit solution of restricted matter configurations
from being interpreted as a general existence-and-uniqueness theorem
for sourced metric reconstruction on the full Kerr exterior.

\section{Discussion and Conclusion}
\label{VIII}

The main purpose of this work was to formulate the closure problem for
first-order, traceful radiation-gauge reconstruction with matter. The metric
perturbation separates into radiative Weyl-curvature, Ricci-driven trace, and
nonradiative mass--angular-momentum completion sectors. We establish local
existence and uniqueness of the triangular transport hierarchy on a regular
Kerr domain once initial data, residual gauge freedom, and completion
parameters are fixed. Proposition~\ref{prop:kerr_closure} specifies the
conservation, boundary, matching, charge, function-space, and kernel
conditions required for a closed global formulation. This is a conditional
criterion; we do not claim generic global invertibility of the Kerr operator.

A conserved thin shell around Schwarzschild provides a benchmark
application. The static sector yields
$\delta M_{\mathrm{ADM}}=\sqrt{f(R)}M_{\mathrm{s}}$, with the Israel relation
adding the binding term $-M_{\mathrm{s}}^2/(2R)$. The axial sector gives
$\delta J=I_{\mathrm{s}}\Omega_{\mathrm{s}}$, where
$I_{\mathrm{s}}=\delta M R(2R-3M)/[3f^2(R)]$, together with rigid interior
frame dragging. Imposing a barotropic surface equation of state, fixed rest
mass, and isotropy closes the leading even-parity rotational deformation.
For the configurations studied, rotational mass corrections are positive,
quadrupoles are negative, and numerical closure residuals remain below
$6.0\times10^{-11}$.

A stationary constant-angular-momentum perfect-fluid torus on Kerr supplies a
smooth conserved nonspherical source. Its Newman--Penrose projections show
that $\Phi_{20}^{(1)}$, $\Phi_{21}^{(1)}$, and $\Phi_{22}^{(1)}$ are generally
excited simultaneously, providing the coupled source data required for the
general Kerr hierarchy.

The present construction has clear limits. The shell is solved explicitly only
through the $(1,0)$, $(1,1)$, and $(1,2)$ sectors; higher mixed orders require
the associated second-order fields. For the torus, the Ricci sources are
explicit and the coupled harmonic boundary-value problem is formulated, but
full metric reconstruction requires solving the remaining Newman--Penrose
equations with compatible horizon and asymptotic data. A stronger global
result would also require operator estimates and a classification of the
zero-frequency gauge and completion kernel.

In summary, we formulate sourced, traceful radiation-gauge reconstruction by
separating its locally well-posed transport system, independent data, residual
kernel freedoms, and conditional global closure requirements. The shell gives
a closed benchmark in which the trace, Coulomb completion, angular momentum,
and rotational observables are determined consistently. The torus provides a
smooth nonspherical source and displays the coupled Ricci structure of the
general Kerr problem. These results establish a systematic framework while
identifying the remaining unproved aspects of global reconstruction.



%

\end{document}